\documentclass[fleq,leqno,11pt]{report}
\usepackage{a4}
\usepackage{euscript}
\usepackage{boxedminipage}
\usepackage{dcolumn}
\usepackage{makeidx} % for generating index
\usepackage{float}
\usepackage{fancyheadings}
\usepackage{fancybox}
\usepackage{rotating}
\usepackage{framed}  % use framed or boites see breaking boxes of text in www.tug.org
\usepackage{appendix}
\usepackage{boites} % use \begin{breakbox} ... \end{breakbox} to put
\usepackage{subfigure}

\usepackage{boxedminipage}
\usepackage{latexsym}
\usepackage{amssymb}
\usepackage{amsmath}
\usepackage{amsfonts}
\usepackage{epsfig}
\usepackage{color}
\usepackage{verbatim}
\usepackage{graphicx}

\newlength{\defbaselineskip}
\newcommand{\setlinespacing}[1]%
           {\setlength{\baselineskip}{#1 \defbaselineskip}}

\newcommand{\setA}{\mathcal{A}}
\newcommand{\setS}{\mathcal{S}}
\newcommand{\setP}{\mathcal{P}}
\newcommand{\setR}{\mathcal{R}}

\setlinespacing{1.2}

\newtheorem{Thm}{Theorem}
\newtheorem{mytheorem}{Theorem}
\newtheorem{Cor}{Corollary}%[chapter]
\newtheorem{Lem}{Lemma}[section]

\newtheorem{Alg}{Algorithm}[chapter]
\newtheorem{Fig}{Figure}[section]

\newtheorem{Exa}{Example}[section]
\newtheorem{Prop}{Proposition}[chapter]
\newtheorem{definition}{Definition}
\newtheorem{lemma}[Thm]{Lemma}

\newcommand{\bc}{\begin{center}}
\newcommand{\ec}{\end{center}}

\newcommand{\ben}{\begin{enumerate}}
\newcommand{\een}{\end{enumerate}}
\newcommand{\beq}{\begin{equation}}
\newcommand{\eeq}{\end{equation}}
\newcommand{\eit}{\end{itemize}}
\newcommand{\bit}{\begin{itemize}}
\newcommand{\bex}{\begin{Exa}}
\newcommand{\eex}{\end{Exa}}
\newcommand{\betab}{\begin{tab}}
\newcommand{\etab}{\end{tab}}
\newcommand{\efig}{\end{Fig}}
\newcommand{\befig}{\begin{Fig}}
\newcommand{\berem}{\begin{rem}}
\newcommand{\erem}{\end{rem}}
\newcommand{\balg}{\begin{Alg}}
\newcommand{\ealg}{\end{Alg}}

\newcommand{\bprop}{\begin{Prop}}
\newcommand{\eprop}{\end{Prop}}
\newcommand{\belem}{\begin{Lem}}
\newcommand{\elem}{\end{Lem}}
\newcommand{\bcor}{\begin{Cor}}
\newcommand{\ecor}{\end{Cor}}
\renewcommand{\epsilon}{\varepsilon}

\newcommand{\argmax}{\mathop{\rm argmax}}

\def\squarebox#1{\hbox to #1{\hfill\vbox to #1{\vfill}}}

\def\acro#1#2{\vskip4pt\hbox to\textwidth{\large
\hbox to5pc{#1\hfill}\vtop{\advance\hsize by
-5pc\raggedright\noindent#2}}}

\def\symbol#1#2{\vskip4pt\hbox to\textwidth{\large
\hbox to5pc{#1\hfill}\vtop{\advance\hsize by
-5pc\raggedright\noindent#2}}}

\begin{document}
%--------------no page numbering ------------------------------------------
\renewcommand{\thepage}{}

%-----------line spacing --------------------------------------------------
%    \mbox{}\\
%    \vspace{4cm}
%---------- title pages ----------------------------------------------------
\begin{center}

    \vspace{32.5 cm}
    \large {\textbf{\ \ }}\\
    \huge {\textbf{T-Plots: A Novel Approach to Network Design}}\\
    \vspace{3.5 cm}
    \vspace{6.5 cm}
    \LARGE{\textbf{ITAMAR COHEN}}\\
    \vspace{3 cm}
\end{center}\hspace{-.2cm}
\newpage
{\ \ }
\newpage

\begin{center}

    \large {\textbf{\ \ }}\\
    \huge {\textbf{T-Plots: A Novel Approach to Network Design}}\\
    \vspace{2.5 cm}
    \Large{RESEARCH THESIS}\\
    \vspace{2.5 cm}
    \Large{In Partial Fulfillment of The
Requirements for the Degree of Master of Science in Electrical Engineering}\\
    \vspace{3.3 cm}
    \LARGE{\textbf{ITAMAR COHEN}}\\
    \vspace{3.3 cm}
\end{center}\hspace{-.2cm}
\begin{center}
\mbox{Submitted to the Senate of the Technion - Israel Institute of
Technology}\\
Kislev, 5768 Nov 2007
\newpage

\mbox{}\\[2cm]
\end{center}
{\ \ }
\newpage

\begin{center}
\begin{minipage}[c]{12cm}
\begin{center}
\Large{\textbf{ The research thesis was done under the supervision
of Dr. Isaac Keslassy in the Faculty of Electrical Engineering.\\[1.5cm] %\vspace{1cm}
        I gratefully thank Dr. Isaac Keslassy for his dedicated supervision, for the continuous support and for the
        inspiring ideas.\\[0.5cm] %\vspace{1cm}
        I would like to thank also the computer networks lab team: Yoram Orchen, Yoram Yihyie and Hai Vortman, for the technical support.\\[5.5cm] %\vspace{1cm}
        The generous financial help of the Banin Fund and of Cisco Israel are gratefully acknowledged.\\[.5cm] %\vspace{1cm}
        }}
\newpage

\end{center}
\end{minipage}
\end{center}
%--------------- table of contents\ ----------------------------------------
% \setlinespacing{.7}
\newpage
{\ \ }
\newpage

\tableofcontents
\newpage
\listoffigures
\newpage

% \setlinespacing{1.2}

\begin{abstract}
\mbox{} \vspace{-1cm} \label{abstract}

It is accepted wisdom that changes in the traffic matrix entail
capacity over-provisioning, but there is no simple measure of just
how much over-provisioning can buy. In this Thesis, we aim to
provide the network designer with a simple view of the network
robustness to traffic matrix changes. We first present the Traffic
Load Distribution Plots, or T-Plots, a class of plots illustrating
the percentage of traffic matrices that can be serviced as a
function of the capacity over-provisioning. For instance, from a simple look at their T-Plots, network designers can guarantee that
their network services all admissible traffic matrices, or 99\% of
permutation traffic matrices, or all traffic matrices with
ingress/egress load at most half the maximum. We further show that, 
unfortunately, in the general case plotting T-Plots is \#P-Complete,
i.e., that it is impossible to plot a T-plot in a polynomial time by
the noon tools. However, we show that T-Plots can sometimes be
closely modeled as Gaussian, thus only using two values (mean and
variance) to quantify the robustness of a capacity allocation to
traffic matrix changes. We further utilize these Gaussian T-Plots to
provide a more robust capacity allocation. Finally, we demonstrate
the benefits of using T-Plots by showing results of extensive
Monte Carlo simulations in a real backbone network.

This Thesis was submitted in 2007. Since then, the results that appeared in it were applied in various networking environments. 
In this newer version, we revisit the results 13 years later and explain their relevance to state-of-the-art problems in network design.
\end{abstract}
%----------- start numbering ---------------------------------------------
       \renewcommand{\thepage}{\arabic{page}}
        \setcounter{page}{1}
%%------------------------------------------------------------------------
    \setlinespacing{1.2}
    \chapter{Introduction} \label{chap:intro}
\section {Overview}
Network design aims to i) guarantee high throughput and low delay
for current and future traffic demands, ii) minimize the amount of
over-provisioned capacity and iii) minimize the number of dynamic
routing changes, which cause undesired effects, such as out-of-order
arrival of packets and drastic changes in traffic flows.

Balancing between these three goals, which often contradict each
other, necessitates the usage of efficient, yet accurate,
measurement tools. However, the development of such tools is a hard
task, due to the frequent changes in traffic demands in
state-of-the-art networks: a routing algorithm, while optimal for a
typical traffic demand, might fare quite poorly as traffic
conditions change.

The goal of this Thesis is to provide a common practical framework
to evaluate and compare routing algorithms in the always-changing
traffic demands environment. In addition, once the routing algorithm
is chosen, the tools developed here enable a simple, yet efficient,
capacity allocation scheme for achieving high throughput without
wasting network resources.

The main contribution of this Thesis is the introduction of the {\em
T-Plots}, or Traffic Load Distribution Plots, a class of plots
illustrating the distribution of the load generated by the whole set
of possible future traffic demands (called{\em T-Set}). The network
operator first defines a routing algorithm, a link capacity
allocation, and a T-Set that reflects the changes in traffic demand.
Then, the operator can plot the load distribution and directly
evaluate the efficiency of a traffic engineering approach according
to some given metric (e.g., average load, worst-case load, or
99\%-cutoff load).

\begin{figure}
\centering
\includegraphics
% [bb=0 920 1200 50,width=8cm, height=6cm]
[width=8cm, height=6cm]
{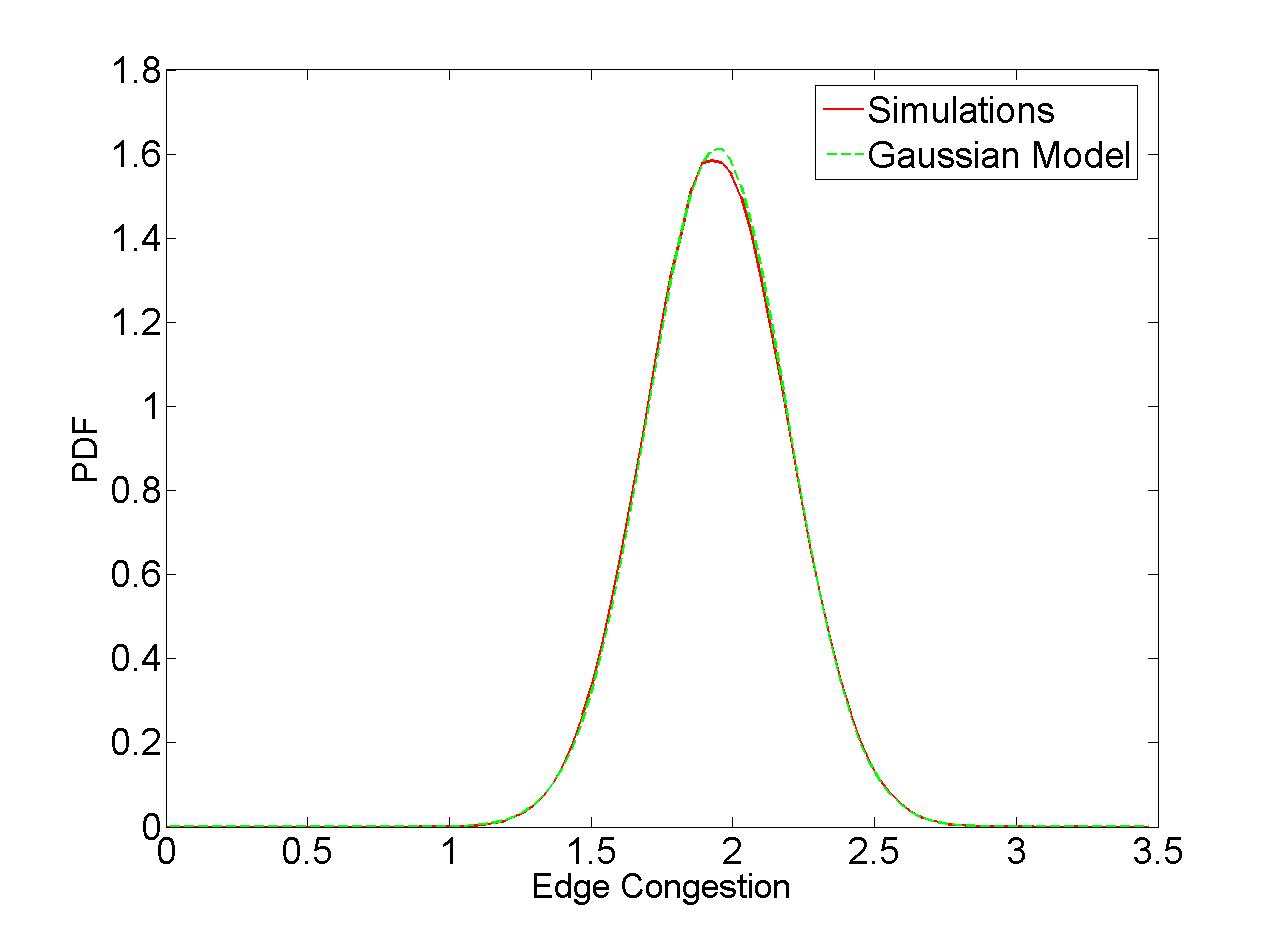}
\caption{Link load distribution on the (Kansas City, Indianapolis)
link in Abilene network, when the T-Set is $\setA$}
\label{fig:Abi_DSS_e13_PDF}
\end{figure}

Figure \ref{fig:Abi_DSS_e13_PDF} provides a T-Plot example. It shows
the distribution of the normalized load on the loaded link in
Abilene backbone network (details in Chapter
~\ref{Sec9:Simulations}) \cite{key-Abi}. The T-Set is assumed to be
the set $\setA$ of all admissible traffic matrices, (see Chapter
~\ref{Sec3:model}). In this T-Plot, the average load is about 1.9,
the 98\%-cutoff load is about 2.5, and the worst-case load can be
found to be exactly 5 (with a negligible density); in other words,
this T-Plot shows that when capacity equals half of the worst-case
load, 98\% of the matrices in the T-Set can already be serviced.
Thus, using this T-Plot, an operator can directly deduce the
performance of its traffic engineering algorithm and obtain clear
performance guarantees. The operator might decide, for instance,
that the marginal benefit of allocating more capacity beyond 2.5 is
not worth the cost. Incidentally, note that the T-Plot can be
closely modeled as Gaussian -- we will later develop on that point.

Depending on the problem faced, T-Plots can be considered in several
ways. They mainly provide a theoretical analysis tool and a unified
view of the different traffic engineering and capacity allocation
algorithms. But they also constitute a practical day-to-day
monitoring tool for network operators, who have a given topology and
a given routing algorithm, and simply want to monitor the influence
of traffic changes on their network performance. The way T-Plots are
used and their associated T-Sets highly depend on the reasons
underlying the changes in traffic demands.

In this Thesis, we set out to study the advantages and drawbacks of
T-Plots. We first demonstrate that the exact computation of T-Plots
is \#P-Complete, thus killing early hopes of easy results. Indeed,
computing exact T-Plots on networks of more than five nodes proves
extremely hard, if not impossible, when using standard tools like
Vinci~\cite{key-Vinci}. However, we propose a simple random-walk
algorithm that provides a close approximation to the exact T-Plot,
thus enabling us to study T-Plots in typical networks. Further, we
also show that T-Plots can sometimes be closely modeled as Gaussian
when using typical regular T-Sets. Therefore, the knowledge of two
single values (mean and variance) is enough to provide a model of
the whole load distribution and quantify the robustness of a
capacity allocation scheme to traffic matrix changes. Further, we
also determine exact simple bounds on the load distribution, thus
enabling operators to provide strict performance guarantees over any
T-Set. Further, we illustrate the possible use of T-Plots in
optimization schemes, by providing a toy model for a capacity
allocation scheme, in which the capacity provided to each link is
equal to the sum of its average load and a multiple of the load
standard deviation. Using a toy model of the Abilene backbone
network, we show that such a simple scheme is surprisingly close to
the performance bounds generated by the set of optimally robust
capacity allocations.

We would like to stress that in our view, the key aspect of this
work is the way it opens the road for future studies. Our results enable the
research community to compare and judge algorithms along the same
scale. Further, we show that there is more to look at in optimization algorithms than the average-case and the worst-case. Finally, we readily admit that it is yet unclear in what exact conditions T-Sets will yield Gaussian T-Plots; and, while we provide some intuition, we leave to future
studies this fundamental question on traffic engineering models.

Some of our results also appear in~\cite{Noc-Conf, Noc-Conf-TR, Noc-Journal}. These results were applied mainly in the context of Network-on-Chip, e.g., in~\cite{chudzikiewicz2013resources}. 
However, over the years similar ideas were extensively studied also in other networking environments. For instance, the emergence of data-center-networks in recent years has increased the interest in developing oblivious routing algorithms with congestion guarantees in such networks~\cite{Beyond-Disco-Balls}. Using a 
statistical approach for tackling uncertainty and heterogeneity in communication networks was studied in the context of buffering and scheduling~\cite{Qmist-Conf, Qmist-Journal}, scheduling~\cite{APSR-Infocom20-poster, APSR-IFIP}, and caching~\cite{Accs-Strategies-Conf, Accs-Strategies-ToN, CAB, FN-Aware-Conf}.

\section {Previous work} \label{sec:previous_work}
\subsection {Dynamic routing changes} \label{sec:dyn} Dynamic routing
algorithms try to perform routing changes which would minimize the probability that a link in the network would become overloaded
\cite{key-PATHNECK,key-MIRA}. However, performing such routing
changes requires full knowledge of the current and future traffic
demands, which are in practice not available \cite{CHANGING_TRAFFIC,
Kodialam, key-VLB, VLB-HOMO}. In addition, routing changes cause
undesired affects, such as out-of-order arrivals of packets and
drastic changes in traffic flows.

\subsection {Oblivious routing algorithms} \label{sec:obl}
%many recent works discuss routing algorithms, which don't employ
%dynamic routing changes (unless necessary, due to network failures)
Due to the practical problems when employing dynamic routing
algorithms, many recent works discuss the optimization of oblivious
routing algorithms, i.e. algorithms in which the routing between
each (source, sink) pair of nodes is determined in advance, and
doesn't change over time. Most works use the congestion in the
network as the optimization criterion, and prove theoretical bounds
on the ratio between the performances that an oblivious routing
algorithm can achieve and those of an optimal adaptive routing
algorithm ~\cite{BOUNDS, OPT-OBL-POLY, NEW-BOUNDS,
PRACTICAL-OPT-OBL-POLY, key-TP_centric}. Other objective functions
are the amount of capacity over-provisioning \cite{key-VLB,
key-TP_centric}, and the length of the routing paths
\cite{STRETCH_N_LB, key-TP_centric}. Other works generalize the
optimization criterion for a general function of the total flow on
each edge in the network, where this function obeys some
requirements (e.g. being concave) \cite{Obl-general,
key-Uncertain_TMs, Zhang1}.

A promising class of network architectures are those based on
Valiant Load Balancing (VLB) \cite{key-Val, key-VLB, Kodialam,
VLB-HOMO}. VLB guarantees a constant throughput, even at traffic
changes. The drawbacks of VLB schemes are that they typically entail
longer routing paths, and require larger network resources.
\cite{VLB_DIRECT} proposes an enhancement to VLB that somewhat
overcomes these drawbacks, but it assumes additional knowledge about
the current traffic in the network, thus hurting the simplicity and
easiness of implementation.

In practice, oblivious routing algorithms may also perform dynamic
routing changes due to failures of links or nodes in the network.
Some works aim to improve the resilience of network architecture,
that deploy oblivious routing algorithms. The objective is to
minimize the amount of over-provisioning used for guaranteeing
restoration from failures \cite{OBL_RESTORATION, PRE-PROVISIONING},
and to minimize the number of dynamic routing changes due to each
link failure \cite{key-COPING}.

Some works suggest hybrid algorithms, in which the baseline routing
is oblivious, but it may change over time, aiming to reduce the
congestion in the network \cite{key-GOAL} or the length of the
routing paths \cite{VLB_DIRECT}.

\subsection {Traffic matrices and T-Sets} \label{sec:T-Sets} A \emph{traffic matrix} is a
non-negative matrix, where the rows represent the sources, and the
columns represent the destinations (formal definitions of all the
terms used in the introduction appear in Chapter \ref{Sec3:model}).

When evaluating an oblivious routing algorithm, one should consider
its performances over the set of allowable traffic matrices (which
we call {\em T-Set}). Several T-Set models have been proposed in the
literature. The chosen T-Set is often the set $\setA$ of all
admissible traffic matrices, i.e. of row/column sums bounded by the
ingress/egress capacity, which are typically normalized to one
(homogeneous case), as popularized by the hose model
~\cite{Duffield,Kodialam,key-Typical_vs_WC, VPNS}. Some works also
focus on the set $\setP$ of permutation
matrices~\cite{key-TP_centric,key-O1TURN,key-GOAL}, since in the
homogeneous case the worst-case admissible traffic matrices are
permutations~\cite{key-WC_traffic}, and every admissible traffic
matrix can be represented as a linear combination of permutation
matrices ~\cite{key-Birkhoff,key-VN}. Other approaches rely on
historical traffic matrix values, for instance, based on one-hour
window observations~\cite{Agarwal, Duffield}, on critical matrices
among those observed~\cite{Zhang3}, on estimated traffic
matrices~\cite{Roughan}, or on the sum of a typical traffic matrix
and small fluctuations~\cite{CHANGING_TRAFFIC}.
~\cite{CHANGING_TRAFFIC} also suggests methods for modeling a
typical traffic matrix. However, the models proposed in
~\cite{CHANGING_TRAFFIC} implicitly assume that the entries in the
traffic matrix are independent. This assumption is used in
traditional network design schemes, which use the pipe model, but it
doesn't hold in networks that deploy the hose model \cite{VPNS,
Duffield}. Finally, one can refer to~\cite{Zhang2} as to the
tradeoffs involved in determining the size of the T-Set.

\subsection {T-Plots}
The reference that is most related to T-Plots in backbone networks
appears to be ~\cite{key-Typical_vs_WC}. \cite{key-Typical_vs_WC}
implicitly discusses the link between the CDF (Cumulative
Distribution Function) of the global congestion and the required
amount of overprovisioning in networks. In interconnection networks,
several studies have presented plots of the throughput PDF
(Probability Density
Function)~\cite{key-O1TURN,key-TP_centric,key-GOAL}, but mainly as a
mean for finding the average-case throughput and without real
interest in the whole distribution. In this Thesis, we extend these
results by modeling and bounding the whole distribution (and not
merely the average-case and worst-case) and showing its high
importance for analyzing the network performances and optimizing the
capacity allocation.

\section{Outline of the Thesis}
This Thesis is structured as follows. We formulate the T-Plot model
in Chapter ~\ref{Sec3:model}, and prove its \#P-completeness in
Chapter ~\ref{Sec4:sharp-P}. Then, in Chapters ~\ref{Sec5:gaussian}
and \ref{Sec6:guarantees}, we provide a Gaussian view of the edge
T-Plots, as well as strict performance, guarantees, and generalize
these results to global T-Plots in
Chapter~\ref{Sec7:Approximation-and-bounds}. Finally, in
Chapter~\ref{Sec8:Improving-the-capacity}, we introduce a simple
capacity allocation scheme, which we evaluate, together with the
other results, in Chapter~\ref{Sec9:Simulations}.

\chapter{T-Plot model}\label{Sec3:model}
We will now introduce the notations used to define T-Sets and
T-Plots.

\textbf{Network --} Consider a directed graph $G(V,E)$ with
\emph{n=$|V|$} nodes and $|E|$ edges. Node \emph{i} may initiate
traffic at a rate up to $r_{i}$, and receive traffic at a rate up to
$q_{i}$. For simplicity of presentation we assume that
$r_{i}=q_{i}=1$ for all \emph{i} (normalized homogeneous case), and
will later explain how to extend this simplistic model to the
general case. Each edge \emph{e} is allocated a positive capacity
$c(e)>0$. We will say that an edge $e$ is a \emph{strictly minimal
edge} if $c(e')>c(e)$ for each edge \emph{e'} different from $e$,
and a \emph{bridge} if removing $e$ would increase the number of
components in the graph.

\textbf{T-Set --} A \emph{traffic matrix} is an $n \times n$
non-negative matrix, where the entry in the location $(i,j)$
represents the amount of traffic from node $i$ to node $j$. We are
interested in analyzing the network performance when traffic
matrices belong to different Traffic Matrix Sets (T-Sets). For
instance, typical T-Sets include:
\begin{itemize}
\item $\setA (n)$, the set of all
the doubly-substochastic traffic matrices, i.e. the $n \times n$
non-negative matrices for which each row and column sum is at most
1:
\begin{equation}
\setA(n)=\left\{ D|\forall i,j: 0 \leq D_{ij} \leq 1,
\sum_{j}D_{ij}\leqslant 1,\sum_{j}D_{ji}\leqslant 1\right\}
\label{eq:setA}
\end{equation}

In this work we use $\setA(n)$ interchangeably with \emph {all
admissible traffic matrices}.
\item $\setS(n)$, the set of doubly-stochastic traffic matrices
(i.e. $n \times n$ non-negative matrices for which the row and
column sums equal exactly 1):
\begin{equation}
\setS(n)=\left\{ D|\forall i,j: 0 \leq D_{ij} \leq 1, \sum_{j}D_{ij}
= 1,\sum_{j}D_{ji} = 1\right\} \label{eq:setS}
\end{equation}

\item $\setP(n)$, the discrete set of $n!$ permutations (i.e. the set of
$n \times n$ 0-1 matrices with exactly single 1 in each row and
column):

\begin{equation}
\setP(n)=\left\{ D|\forall i,j: D_{ij}\in \{0, 1\}, \sum_{j}D_{ij} =
1,\sum_{j}D_{ji} = 1\right\} \label{eq:setS2}
\end{equation}

\end{itemize}

$\setA(n)$, which is defined by a set of linear inequalities, forms
a convex polytope, also known as {}``the routing polytope''
\cite{key-Uncertain_TMs}. $\setS(n)$, which is defined by a set of
linear equalities, forms the faces of the polytope $\setA(n)$ (it is
easy to think of it as the {\em surface} around $\setA(n)$).
$\setP(n)$ is the discrete set of the vertices of this polytope
\cite{key-vol_of_poly_of_SD}. The worst-case traffic matrices in
$\setA(n)$ are found within $\setP(n)$ ~\cite{key-WC_traffic}.

When there is no ambiguity, we shall use $\setA$, $\setS$ and $\setP$ instead of
$\setA(n)$, $\setS(n)$ and $\setP(n)$, respectively.

In the remainder, we will assume that traffic matrices are selected
u.a.r. (uniformly at random) from the T-Set, but this can of course
easily be extended to weighted T-Sets by using a weighted metric
over the T-Set metric space, so that the probabilities of choosing
each traffic matrix correspond to these weights.
%When there is no ambiguity, $\setA(n)$, $\setS(n)$ and $\setP(n)$
%will be denoted as $\setA$, $\setS$ and $\setP$ respectively.

\textbf{Routing --} A \emph{routing} is classically defined as a set
of $n^{2}|E|$ variables $\{f_{ij}(e)\}$, where $f_{ij}(e)$ denotes
the fraction of the traffic from node $i$ to node $j$ that is routed
through edge $e$. The routing is assumed to satisfy the classical
linear flow conservation constraints \cite{OPT-OBL-POLY}. Such
a routing is oblivious in the sense that the routing variables are
independent of the current traffic matrix. Note that the routing is
allowed to depend on both the source and the destination, as in
MPLS, and therefore subsumes routings that only depend on the
destination, as in many Interior Gateway
Protocols~\cite{Zhang1, Zhang2}.

\textbf{Congestion --} The total \emph{flow} crossing an edge $e\in
E$ when routing the traffic matrix $D$ is
$F(e,f,D)=\sum_{i,j}D_{ij}f_{ij}(e)$, where $D_{ij}$ is the
$(i,j)^\textrm{th}$ element of matrix $D$. Then, the \emph{edge
congestion} on edge $e$ is equal to the total flow crossing it
divided by the edge capacity, i.e.
\begin{equation} EC(e,f,D)=\frac{F(e,f,D)}{c(e)}=\frac{\sum_{i,j}D_{ij}f_{ij}(e)}{c(e)}\end{equation}
Thus, when the edge congestion on $e$ is at least 1, the flow
crossing $e$ is not below its capacity, and we will say that $e$ is
\emph{saturated}. Further, a network is saturated if at least one
edge in it is saturated. The \emph{global congestion} of routing $D$
using $f$ will be obtained by maximizing the edge congestion over
all the edges, that is:
\begin{equation} GC(f,D)=\max_{e\in E}\left\{ EC(e,f,D)\right\}
\label{eq:GC}\end{equation} In this thesis, congestion and load are
used interchangeably.

For a saturated network, the \emph{throughput} is defined as the
inverse of the global congestion, and is otherwise made not to
exceed 100\%:
\begin{equation} TP(f,D)=\textrm{min}\{
GC(f,D)^{-1},1\}\label{eq:TP}\end{equation}

\textbf{T-Plot --} T-Plots (the Traffic Load Distribution Plots) can
be classified into two categories: {\em edge T-Plots}, which show
the distribution of an edge congestion generated by traffic matrices
in the T-Set; and {\em global T-Plots}, which similarly show the
distribution of the global congestion generated by traffic matrices
in the T-Set. T-Plots can be represented in a CDF (Cumulative
Distribution Function) or PDF (Probability Distribution Function)
way. For example, the value of the \emph{edge T-plot CDF} at point
\emph{L} is the probability that the edge congestion imposed on that
edge by a traffic matrix selected u.a.r. from the T-Set would be at
most \emph{L}:
\begin{equation}
EC_{CDF}^{T}(e,f,L)=\Pr\left\{ EC(e,f,D)\leq L | D \in T\right\}
\label{eq:def_EC_CDF}\end{equation} Likewise, it it possible to use
the \emph{edge T-plot PDF}, which corresponds to the derivative of
the CDF whenever it exists. (For simplicity, we shall use the same
terms for both discrete and continuous T-Sets.)  The CDF and PDF of
the \emph{global congestion} are, of course, defined similarly.

Given an edge T-Plot, it is easy to find the \emph{worst-case edge
congestion}, as well as statistics such as the \emph{average-case
edge congestion} and the \emph{variance} of the edge congestion,
which are computed over the whole T-Set using a uniform measure on
that metric space.

\chapter{T-Plots are \#P-Complete}\label{Sec4:sharp-P}

Now that we have introduced T-Plots, we will prove that their
computation is \emph{\#P-C}. We will first show it for {\em edge
T-Plots}, which display the congestion distribution on a given edge,
and then as well for {\em global T-Plots}, which plot the
distribution of the global congestion. In both cases, we reduce the
problem of computing the permanent of a 0-1 matrix, which is known
for being \#P-complete, to a specific T-Plot computation problem
with a specific T-Set. Thus, by showing that a specific T-Plot
computation problem is \#P-complete, we demonstrate that the general
computation of an arbitrary T-Plot is \#P-complete as well.

\section {Edge T-Plots are \#P-Complete}\label{edge-sharp-p}
Before proving that the computation of edge T-Plots is \#P-Complete,
we need some preliminaries\cite{key-01Perm}.

\begin{definition}
Let $f$ be a function. We say that \emph{f}$\in$\emph{\#P} if there
exists a binary relation $R$ s.t.
\end{definition}
\begin{itemize}
\item If $(x,y)$$\in R$ then the length of $y$ is polynomial in the length
of $x$.
\item It can be verified in polynomial time that a pair $(x,y)$ is in $R$.
\item For every $x\in\Sigma^{*}$ (the set of all 0-1 strings), $f(x)=|\{ y:(x,y)\in R\}|$
\end{itemize}
\begin{definition}
Given two functions $f,g$
\end{definition}
\begin{itemize}
\item There is a \emph{polynomial Turing-reduction} from $g$
to $f$ (and denote $g\propto f$) if the function $g$ can be computed
in polynomial time using an oracle to $f$.
\item A function $f:\Sigma^{*}\rightarrow N$ is \emph{\#P-Hard} if for
every \emph{g}$\in$\emph{\#P} there is a polynomial reduction
$g\propto f$.
\item A function $f$ is \emph{\#P-Complete} if it is both \emph{\#P-Hard}
and in \emph{\#P}.
\end{itemize}
\begin{definition}
Given an $n \times n$ matrix $A$, the permanent of $A$ is defined as
\begin{equation}
Perm(A)=\sum_{\sigma\in
\setP(n)}\prod_{i=1}^{n}A_{i,\sigma(i)}\label{eq:perm-def}
\end{equation}
\end{definition}
We denote the problem of computing the permanent of a 01-matrix as
\emph{01-Perm}. It was shown in \cite{key-01Perm} that 01-Perm is
\#P-Complete.
\begin{lemma}
Let $G(V,E)$ be a connected graph. Let $e$ be a non-bridge edge.
Then, for each pair of nodes $(s,d)$, there exists at least one
walk, which uses \emph{e} exactly once, and at least one additional
walk, which doesn't use \emph{e}.
\end{lemma}
\begin{proof}
Let's denote $e=(a,b)$. Since $e$ is not a bridge, for each (source,
destination) pair of nodes $(s,d)$, there exists a walk that doesn't
use $e$. Let's also construct a walk from $s$ to $d$ that uses $e$
exactly once. First, from $s$ to $a$, we use a walk that doesn't use
$e$ (such a walk exists, because $e$ is not a bridge) Then, from $a$
to $b$, we cross through $e$. Finally, from $b$ to $d$, we use
another walk that doesn't use $e$ (again, such a walk exists,
because $e$ is not a bridge).
\end{proof}

We will Theorem 1 by showing that the problem is \emph{\#P}
(Lemma 2), and then showing a polynomial-time reduction from the
problem of \emph{01-Perm} (Lemma 3).

Intuitively, given an $n \times n$ matrix \emph{A} and a non-bridge
edge \emph{e}, this reduction computes a routing algorithm \emph{f}
that routes packets from source \emph{s} to destination \emph{d} via
\emph{e} iff $A_{sd}=1$. We will show that as a consequence, the
value of the permanent of \emph{A} is equal to the number of
permutations, for which all the flow is routed via \emph{e}.
Therefore, using this equality, an oracle to a specific point of the
PDF of the congestion on \emph{e} suffices for calculating the
permanent of \emph{A}.
\begin{lemma}
\label{lem:Them1 is sharp-P}Let $G(V,E)$ be a directed graph, in
which the traffic is routed according to an oblivious routing
algorithm $f$. Let \emph{e} be a non-bridge edge. Then,
$EC_{PDF}^{\setP}(e,f,L)$ is in \emph{\#P}.
\end{lemma}
\begin{proof}
Let us denote $L=\frac{n}{c(e)}$. Let's define the binary relation
$R$ as follows: \begin{equation} ((e,f,L),\sigma)\in
R\Leftrightarrow EC(e,f,\sigma)=L\label{eq:def of L}\end{equation}
\begin{itemize}
\item The size of the representation of a permutation is polynomial in the
representation of $(e,f)$.
\item It can be verified in polynomial time whether $((e,f,L),\sigma)\in R$
by calculating $EC(e,f,\sigma)$.
\item As $EC_{PDF}^{\setP}(e,f,L)$ represents the probability of a randomly
chosen permutation to impose congestion \emph{L} on edge \emph{e},
and $|\setP(n)|=n!$, \begin{equation}
n!EC_{PDF}^{\setP}(e,f,L)=|\{\sigma\in
\setP(n):EC(e,f,\sigma)=L\}|\label{eq:n_PDF}\end{equation}
\end{itemize}
Note that the multiplication by $n!$ doesn't affect
$EC_{PDF}^{\setP}(e,f,L)$ being in \emph{\#P.}
\end{proof}
\begin{lemma}
\label{lem:PI_eq_EC}Let \emph{A} be an $n \times n$ 01-matrix. Let
$\sigma$ be a permutation of $\{1,2,...,n\}$. Then, it is possible
to construct in polynomial time a routing algorithm \emph{f} s.t.
\begin{equation}
\prod_{i=1}^{n}A_{i,\sigma(i)}=1\Leftrightarrow
EC(e,f,\sigma)=L\label{eq:lem 3}\end{equation}
\end{lemma}
\begin{proof}
We will construct an oblivious routing algorithm $f$ s.t.
\begin{equation} \forall i,j\in
V:f_{ij}(e)=A_{ij}\label{eq:Aij=f_ij}\end{equation}
 In other words, if $f_{sd}(e)=1$, $f$ will use \emph{e} exactly
once when routing packets from \emph{s} to \emph{d}. Otherwise,
\emph{f} will not use \emph{e} when routing packets from \emph{s} to
\emph{d}. By the construction:
% $ Change from v2.3:
% changed from A_{ij}(e) to A_{ij} at the first line of the eqnarray below.
\begin{eqnarray}
\prod_{i=1}^{n}A_{i,\sigma(i)}=1 & \Leftrightarrow & \sum_{ij}\sigma_{ij}A_{ij}=n\nonumber \\
 & \Leftrightarrow & \sum_{ij}\sigma_{ij}f_{ij}(e)=n\nonumber \\
 & \Leftrightarrow & EC(e,f,\sigma)=\frac{n}{c(e)}=L\label{eq:proof of PI <=> EC}\end{eqnarray}
If $A_{sd}=0$, \emph{f} routes the packets from $s$ to $d$ via a
walk, which does not use $e$. Else, $f$ routes the packets from $s$
to $d$ via a walk which uses \emph{e} exactly once, as described in
the proof of Lemma 1. \emph{f} can be calculated by running the
Dijkstra Algorithm {[}3] $n$ times on the graph $G'(V,E\backslash\{
e\})$, each time taking a different node in the graph as the source
node. Each edge in \emph{E\textbackslash{}\{e\}} has a unit weight.
The complexity of the construction is polynomial, as it requires
$O(|V|)$ runs of the Dijkstra Algorithm, where each run takes
polynomial time.
\end{proof}
\begin{mytheorem}
When the T-Set is the set of permutations $\setP$, finding the edge
T-Plot of a non-bridge edge \emph{e} is \emph{\#P-C}.
\end{mytheorem}
\begin{proof}
Let \emph{A} be a $n \times n$ 01-matrix. By Lemma \ref{lem:Them1 is
sharp-P}, $EC_{PDF}^{\setP}(e,f,L)$ is \emph{\#P.} Successively
using Lemma \ref{lem:PI_eq_EC} and (\ref{eq:n_PDF}):\begin{eqnarray}
Perm(A) & = & |\{\sigma\in \setP(n):\prod_{i=1}^{n}A_{i,\sigma(i)}=1\}|\nonumber \\
 & = & |\{\sigma\in \setP(n): EC(e,f,\sigma)=L\}|\label{eq:Thm1_proof}\\
 & = & n!EC_{PDF}^{\setP}(e,f,L)\nonumber \end{eqnarray}
 This is a polynomial reduction from \emph{01Perm} to $EC_{PDF}^{\setP}$.
As \emph{01Perm} is \#P-Hard, $EC_{PDF}^{\setP}$ is \#P-C, where
even the task of computing the value of the distribution at an
arbitrary point is \#P-C. Consequently, $EC_{CDF}^{\setP}$ is \#P-C
as well.
\end{proof}
\begin{Cor}
In the general case, finding the edge T-Plot is \emph{\#P-C}.
\end{Cor}

\section {Global T-Plots are \#P-Complete}
We have just proved that finding the edge T-Plot is \emph{\#P-C}. We
shall now prove that finding the global T-Plot is \emph{\#P-C} as
well.

In a network with a strictly minimal edge $e_{m}$, the worst-case
congestion on $e_{m}$ is higher than the worst-case congestion on
any other edge in the network. Thus, $e_{m}$'s worst-case edge
congestion is equal to the network's \emph{global} congestion. Thus,
similar construction to that used in Section ~\ref{edge-sharp-p}
would suffice to prove that \emph{$GC_{CDF}^{\setP}$} is
\#\emph{P-C} as well.

\begin{mytheorem}
When the T-Set is the set of permutations $\setP$, finding the
global T-Plot of a graph that includes a strictly minimal edge is
\emph{\#P-C}.
\end{mytheorem}
\begin{proof}
Let us denote: $L=\frac{n}{c(e_{m})}$. Let \emph{A} be a $n \times
n$ 01-matrix. $GC(f,L)\in$ \emph{\#P} by the same proof as that of
Lemma \ref{lem:Them1 is sharp-P}. By Lemma \ref{lem:PI_eq_EC}, there
exists a polynomial-time construction of a routing algorithm
\emph{f} s.t.
\begin{equation}
\prod_{i=1}^{n}A_{i,\sigma(i)}=1\Leftrightarrow
EC(e_{m},f,\sigma)=L\label{eq:App_B_1}\end{equation}
 As $e_{m}$ is strictly minimal, its worst-case congestion
is higher than the worst-case congestion on every other edge, and
therefore its worst-case edge congestion is equal to the network's
\emph{global} congestion. We finally get:
\begin{equation} \forall\sigma\in \setP(n),\forall e\in E,e\neq
e_{m}:EC(e,f,\sigma)\leq\frac{n}{c(e)}<L\label{eq:App_B_2}\end{equation}
 And thus $e_{m}$'s worst-case edge congestion is equal to the network's
\emph{global} congestion
\begin{equation}
GC(f,\sigma)=L\Leftrightarrow EC(e_{m},f,\sigma)=L\label{eq:GC <> EC
(them 2)}\end{equation} By combining Equations (\ref{eq:Thm1_proof})
and (\ref{eq:GC <> EC (them 2)}):
\[
Perm(A)=n!EC_{PDF}^{\setP}(e_{m},f,L)=n!GC_{PDF}^{\setP}(f,L)\] As
in the proof of Theorem 1, even the task of computing the value of
$GC_{PDF}^{\setP}$ at an arbitrary point is \#P-C. Consequently,
$GC_{CDF}^{\setP}$ is \#P-C as well.
\end{proof}

\begin{Cor}
In the general case, finding the global T-Plot is \emph{\#P-C}.
\end{Cor}

We have proved that the exact computation of T-Plots is
\#P-Complete, regarding local congestion as well as global
congestion. This is reflected by the extreme complexity in computing
exact T-Plots for networks of reasonable size. For instance, given a
maximum algorithm running time of one day, we could not compute
exact T-Plots even on a simple 3-by-2 network when using
Vinci~\cite{key-Vinci}, a standard convex-polytope volume
computation tool. An underlying reason behind this complexity is
that computing the T-Plot is equivalent to multiple polytope
volume computations
--- this number even being infinite in the case of continuous T-Sets ---
and polytope volume computations are \emph{\#P-C} in the general
case~\cite{key-Vol_is_P-C}. Therefore, even getting exact values of
T-Plots at fixed bins proves extremely hard. In fact, as the space
dimension grows, \cite{Elekes} shows that we cannot hope for a good
approximation when using a deterministic algorithm unless we take
exponentially many points, and the best possible approximation of
any polynomial-time algorithm is at least exponential in the space
dimension. T-Plot computation is made even harder by the fact that
the number of dimensions in the volume computation is equal to the
number of flows, hence growing like $O(n^2)$ and not like $O(n)$.
Consequently, there is not much to be done to make exact computation
easier. Since exact T-Plot computation proves elusive, we can only
try to approximate or bound it. This will be a recurring theme in
this Thesis.

\chapter{Gaussian model}\label{Sec5:gaussian}
We just proved that in the general case, it is \emph{\#P-C} to find
the edge T-Plots. Therefore, we will strive to look for a good
approximation. We will now show that edge T-Plots can sometimes be
closely modeled as Gaussian, so that it suffices to compute two
single values (average and variance) to approximate them. We will
first present how the average and the variance of the edge
congestion can be found, and then study the Gaussian approximation
based on these two values.

\section{Finding the Gaussian
parameters}\label{SubSec:avg-variance}

Our objective is to compute the Gaussian parameters of the T-Plot of
a given edge over a given T-Set. Later in this chapter, we will show
that the general computation of these parameters can often be done
by pre-computing once several T-Set-dependent values (e.g., average
value for flow $(i,j)$), and then plugging in those fixed values
each time simulation parameters are changed (e.g., a new routing
function). Moreover, in several regular T-Sets, this computation is
much simplified. We will present below the results of this
computation when the chosen T-Set is the set of permutations
$\setP$. In that case, the average-case edge congestion on edge $e$
using routing $f$ is:
\begin{eqnarray}
EC_{ac}^{\setP}(e,f) & = & \frac{1}{n!}\sum_{\sigma\in \setP}EC(e,f,\sigma)\nonumber \\
 & = & \frac{1}{n!}\sum_{\sigma\in \setP}  \frac{\sum_{ij}\sigma_{ij}f_{ij}(e)}{c(e)}\nonumber \\
 & = & \frac{1}{n!c(e)}\sum_{ij}f_{ij}(e)\sum_{\sigma\in \setP}\sigma_{ij}\nonumber \\
 & = & \frac{1}{nc(e)}\sum_{ij}f_{ij}(e),\label{eq:perm avg}\end{eqnarray}
where the last equality relies on the fact that a given flow $(i,j)$
is only used in ${\frac{1}{n}}^\textrm{th}$ of the permutations.
Likewise, the variance is calculated using the formula
\begin{equation}
 Var_{\sigma \in
\setP}[EC(e,f,\sigma)]=E[EC(e,f,\sigma)^{2}]-E^{2}[EC(e,f,\sigma)],\label{eq:form_for_var}\end{equation}
with
\begin{equation}
E[EC(e,f,\sigma)^{2}] =
\frac{1}{n!c(e)^2}\sum_{ijkl}f_{ij}(e)f_{kl}(e) \left(
\sum_{\sigma\in \setP}\sigma_{ij}\sigma_{kl} \right)
\label{eq:perms_var-EC},\end{equation} where the expectations are
with respect to the random variable $\sigma \in \setP$.
%
%\begin{eqnarray}
%E_{\setP}[L^{2}] & = & \frac{1}{n!}\sum_{\sigma\in \setP}[EC(e,f,\sigma)]^{2}\nonumber \\
%% & = & \frac{1}{n!c(e)}\sum_{\sigma\in \setP}\sum_{ijkl}\sigma_{ij}\sigma_{kl}f_{ij}(e)f_{kl}(e)\nonumber \\
% & = & \frac{1}{n!c(e)}\sum_{ijkl}f_{ij}(e)f_{kl}(e) \left( \sum_{\sigma\in \setP}\sigma_{ij}\sigma_{kl} \right) \label{eq:perms_var-EC}\end{eqnarray}
By basic combinatorial considerations:
\begin{equation}
\sum_{\sigma\in \setP}\sigma_{ij}\sigma_{kl}=\left\{
\begin{array}{cc}
(n-1)! & i=k\wedge j=l\\
(n-2)! & i\neq k\wedge j\neq l \\
0 & (i=k\wedge j\neq l)\vee(i\neq k\wedge j=l)
\end{array}\right.\label{eq:acg P cases}\end{equation}
Therefore, it is possible to plug in these fixed values and directly
compute the exact mean and variance for any edge $e$, capacity
allocation $c(e)$ and routing algorithm $f$.

\section{Generalization to different T-Sets}\label{SubSec:avg-variance-more}
We will present below methods for computing the average and the
variance of the load for various T-Sets.

\subsection{Continuous T-Sets}\label{AppendixC}
We are now interested in demonstrating how the Gaussian parameters
can be computed in typical continuous T-Sets. We will provide
examples when the T-Sets are $\setS$ or $\setA$, using the same
techniques as in Section ~\ref{SubSec:avg-variance}.

First, the average-case edge congestion is the average congestion
caused by a matrix $D$ selected u.a.r. in the T-Set $T$:
\begin{equation}
EC_{ac}^{T}(e,f)=
%E_{D \in T}[EC(e,f,D)] =
\frac{\int_{T}EC(e,f,D)ds}{\int_{T}1ds},\label{eq:EC_ac_DS_in_app}\end{equation}
with the simplified formula:

\begin{eqnarray}
 & \int_{T}EC(e,f,D)ds & =\frac {1}{c(e)}\int_{T}\sum_{ij}\left[D{}_{ij}f_{ij}(e)\right]ds\nonumber \\
 &  & =\frac {1}{c(e)}\sum_{ij}f_{ij}(e)\int_{T}D_{ij}ds\label{eq:faces-ac2}
 \end{eqnarray}
\noindent \begin{flushleft}
\par\end{flushleft} \noindent

In the T-Set $\setS$, the value of $\int_{T}D_{ij}ds$ is independent
of \emph{i,j}. There exist $n^{2}$ different pairs of $(i,j)$ and
$\sum_{ij}\int_{\setS}D{}_{ij}ds = \int_{\setS}\sum_{ij}D_{ij}ds
=\int_{\setS}nds = n\int_{\setS}1ds.$ Hence $
\int_{\setS}D_{ij}ds=\frac{1}{n}\int_{\setS}1ds$, and the average
congestion is finally simply
$EC_{ac}^{T}(e,f)=\frac{1}{nc(e)}\sum_{ij}f_{ij}(e)$.

When the T-Set is $\setA$, it is possible to predetermine the value
of $\int_{T}D_{ij}ds$ by Monte Carlo simulations with an arbitrary
small error, as explained in Chapter ~\ref{Sec9:Simulations}. As
this value depends only on \emph{n} (and not on the topology or the
routing algorithm), it is sufficient to precompute at most one such
simulation for each qualified network.

The variance of the edge congestion is calculated as in Equations
(\ref{eq:form_for_var}) and (\ref{eq:perms_var-EC}):
\begin{equation}
E_{T}[EC(e,f,D)]=\frac{1}{nc(e)^2}
\sum_{ijkl}f_{ij}(e)f_{kl}(e)\int_{T}D_{ij}D_{kl}ds\label{eq:faces_EL2}\end{equation}
By separating the mutual relations of \emph{i,j,k,l} to the same 3
cases as in (\ref{eq:acg P cases}), it is possible to predetermine
the value of $\int_{T}D_{ij}D_{kl}ds$ by Monte Carlo simulations -
again, using at most one simulation for each qualified network.

\subsection{Traffic
matrices with diagonal 0}\label{AppendixD} For the sake of
completeness, please note that it is possible to add the assumption
that a node cannot send traffic to itself, i.e. the entries of the
traffic matrix diagonal will be 0. (Some works in the literature
assume it \cite{key-Typical_vs_WC,OPT-OBL-POLY}, while others do not
\cite{key-WC_traffic,key-O1TURN,key-TP_centric,key-GOAL}.) We denote
by $\setP_{d}$, $\setS_{d}$ and $\setA{}_{d}$ the 0-diagonal subsets
of $\setP$, $\setS$ and $\setA$, respectively.

The \emph{average edge congestion} when T-Set is $\setP_{d}$ is

\begin{equation}
EC_{ac}^{\setP_{d}}(e,f)=\frac{1}{(n-1)c(e)}\sum_{ij}f_{ij}(e)\label{eq:AC_Pd}\end{equation}
For calculating the variance of the edge congestion we will use
again Equation (\ref{eq:form_for_var}), where the value of
$E_{\setP_{d}}[L^{2}]$ is computed using Equation
(\ref{eq:perms_var-EC}):

\begin{equation}
E_{\setP_{d}}[L^{2}]=\frac{n}{|\setP_{d}(n)|*c(e)}\sum_{ijkl}f_{ij}(e)f_{kl}(e)\sum_{\sigma\in
\setP_{d}}\sigma_{ij}\sigma_{kl}\label{eq:E_l2_derang}\end{equation}
 By basic combinatorial considerations:

\[
\frac{n}{|\setP_{d}(n)|}\sum_{\sigma\in
\setP_{d}}\sigma_{ij}\sigma_{kl}=\]

\begin{equation}
\left\{ \begin{array}{cc}
\frac{1}{n-1} & i=k\wedge j=l\\
0 & (i=k\wedge j\neq l)\vee(i\neq k\wedge j=l)\\
\frac{1}{(n-1)(n-2)} & i\neq k\wedge j\neq
l\end{array}\right.\label{eq:E_ld2_derang_cases}\end{equation}
 When the T-Set is $\setS_{d}$ or $\setA_{d}$, the average and the variance
of the edge congestion may be calculated using Equations
(\ref{eq:form_for_var}) and (\ref{eq:faces-ac2}). Note that

\[
\int_{\setS_{d}}EC(e,f,D)ds=\frac{1}{n-1}\sum_{ij}f_{ij}(e)\]
 The rest of the integrals may be calculated by Monte Carlo simulations,
as detailed in Section ~\ref{AppendixC}.

\subsection{Heterogeneous T-Sets}\label{AppendixE}
In our model, node \emph{i} may initiate
traffic at rate up to $r_{i}$, and receive traffic at rate up to
$q_{i}$. At {\em heterogeneous} traffic models, there exists $i$
s.t. $r_{i}\neq q_{i}$, or there exist $i,j$ s.t. $r_{i}\neq r_{j}$
or $q_{i}\neq q_{j}$.

A general heterogenous T-Set $H$ is a convex polytope:

\begin{equation}
H(n)=\left\{ D|\forall i:D_{ij}\geq 0, \sum_{j}D_{ij}\leqslant
r_{i},\sum_{j}D_{ji}\leqslant q_{i}\right\}
\label{eq:H}\end{equation} Note that the \emph{surface} of this
polytope creates an additional T-Set, in which the inequalities in
(\ref{eq:H}) are replaced by equalities.

Let $D$ be a traffic matrix in a heterogeneous T-Set $H$. The
average and the variance of the edge congestion in $D$ may be
calculated by methods similar to those detailed in Section
~\ref{AppendixC}. Note that as the values of the entries in $D$ are
not identically distributed anymore, using Equation
(\ref{eq:EC_ac_DS_in_app}) requires at the worst-case calculating
the integral $\int_{T}D_{ij}ds$ separately for each pair $(i,j)$,
i.e. $n^{2}$ times. Using Equation (\ref{eq:faces_EL2}) requires
calculating up to $n^{4}$ integrals - one for each chosen 4 indices
$(i,j,k,l)\in[1,n]^{4}$. However, recall that it is sufficient to
calculate each of these integrals at most once for each qualified
network.

\section{Gaussian model example}\label{Sec:Gaussian_Model}

In order to study the Gaussian model, we now present two typical
examples of T-Plots of the normalized link congestion on our
simplified model of the Abilene backbone network (details in
Chapter~\ref{Sec9:Simulations}). The T-Set is assumed to be the set
$\setA$ of all admissible traffic demand matrices.

Consider Figure~\ref{fig:Abi_DSS_e13_PDF}, already mentioned in the
Introduction. Figure~\ref{fig:Abi_DSS_e13_PDF} shows the PDF of the
edge congestion on the most-loaded Abilene edge, $e_{13}$, which
corresponds to the (Kansas City, Indianapolis) directed edge (by
symmetry, there is the same T-Plot for the reverse direction of the
link). As can be seen, the fit of the Gaussian model is excellent.
(As explained in Sections ~\ref{AppendixC} and
~\ref{sub:T-Set-representation}, the Gaussian plot is generated
using an independent random-walk sampling algorithm to find the mean
and variance; running more samples or taking the mean and variance
from the simulations plot would simply make the two plots
undistinguishable).

\begin{figure}
\centering
\includegraphics
% [bb=0 920 1200 50,width=8cm, height=6cm]
% [width=8cm, height=6cm]
[width = 0.66\columnwidth]
{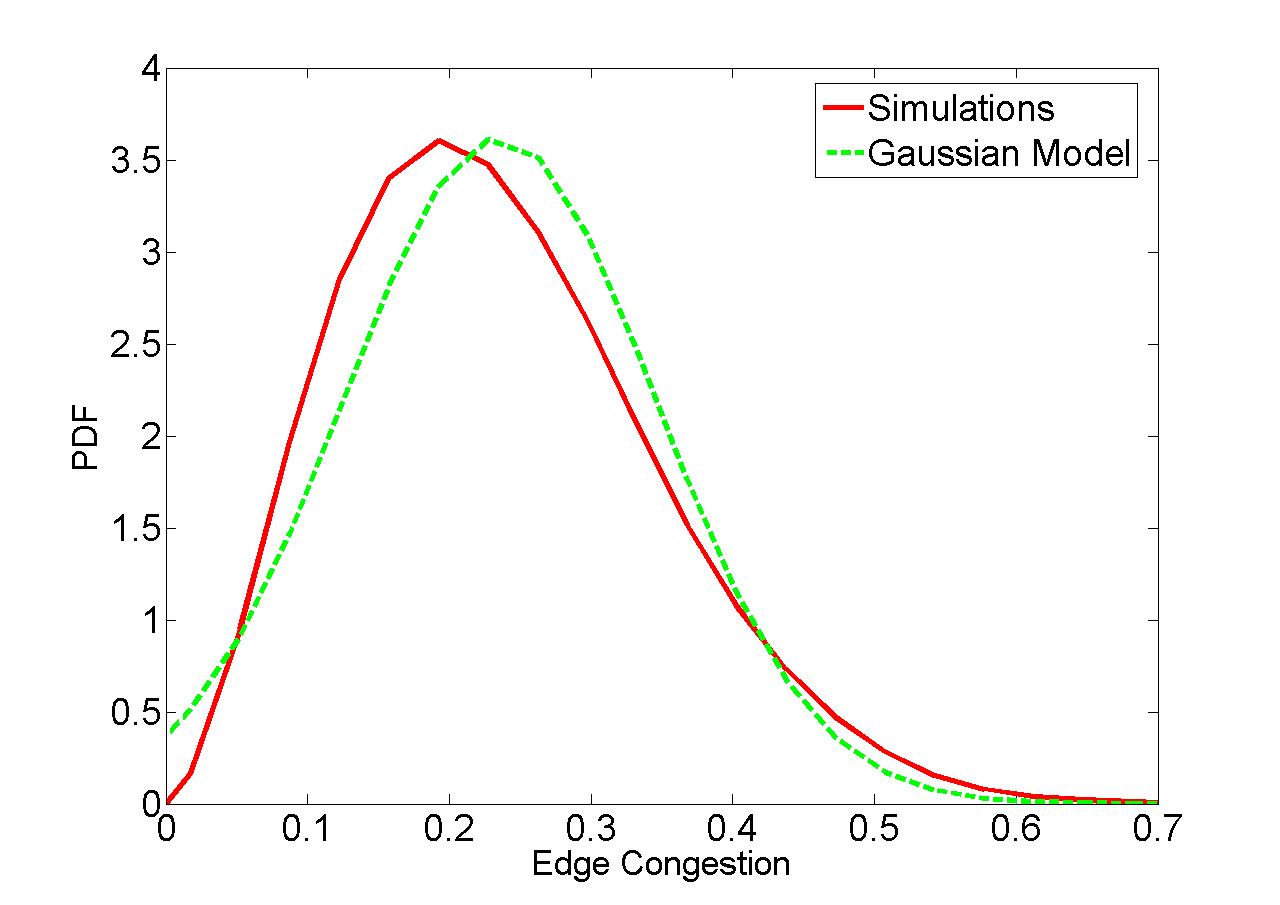}
\caption{Link load distribution on the (Seattle, Sunnyvale) link in
Abilene network, when the T-Set is $\setA$}
\label{fig:Abi_vol_e1_PDF}
\end{figure}

Now consider Figure \ref{fig:Abi_vol_e1_PDF}, a plot of the edge
congestion PDF on the {\em least-loaded} Abilene edge $e_1$,
corresponding to the (Seattle, Sunnyvale) link. Here the Gaussian
model fares poorer. Note that it would also not work much
better if we were to compensate for the hidden negative congestion
values in the left of the Gaussian distribution by multiplying the
positive right side of the PDF by some constant. This example shows
that there is simply no guarantee that the Gaussian approximation
will always work.

\begin{figure}
	\centering
    \subfigure[edge $e_{13}$] {
        \includegraphics
        [width = 0.66\columnwidth]
        % [bb=0 920 1200 50,width=7cm, height=5cm]
        {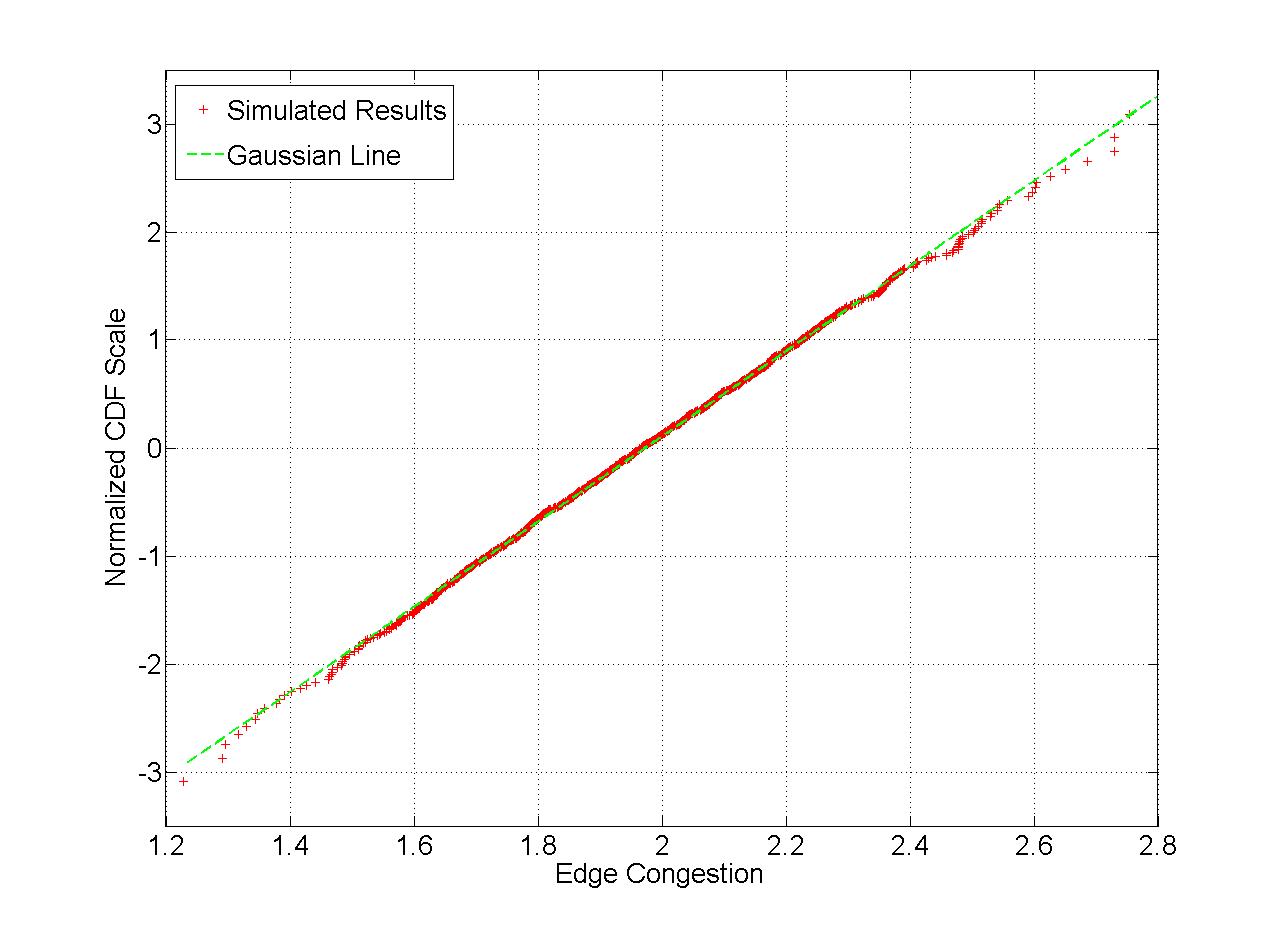}
        \label{fig:Abi_vol_e13_norm_test}}
    \subfigure[edge $e_1$]{
        \includegraphics
        [width = 0.66\columnwidth]
        % [bb=0 920 1200 50,width=7cm, height=5cm]
        {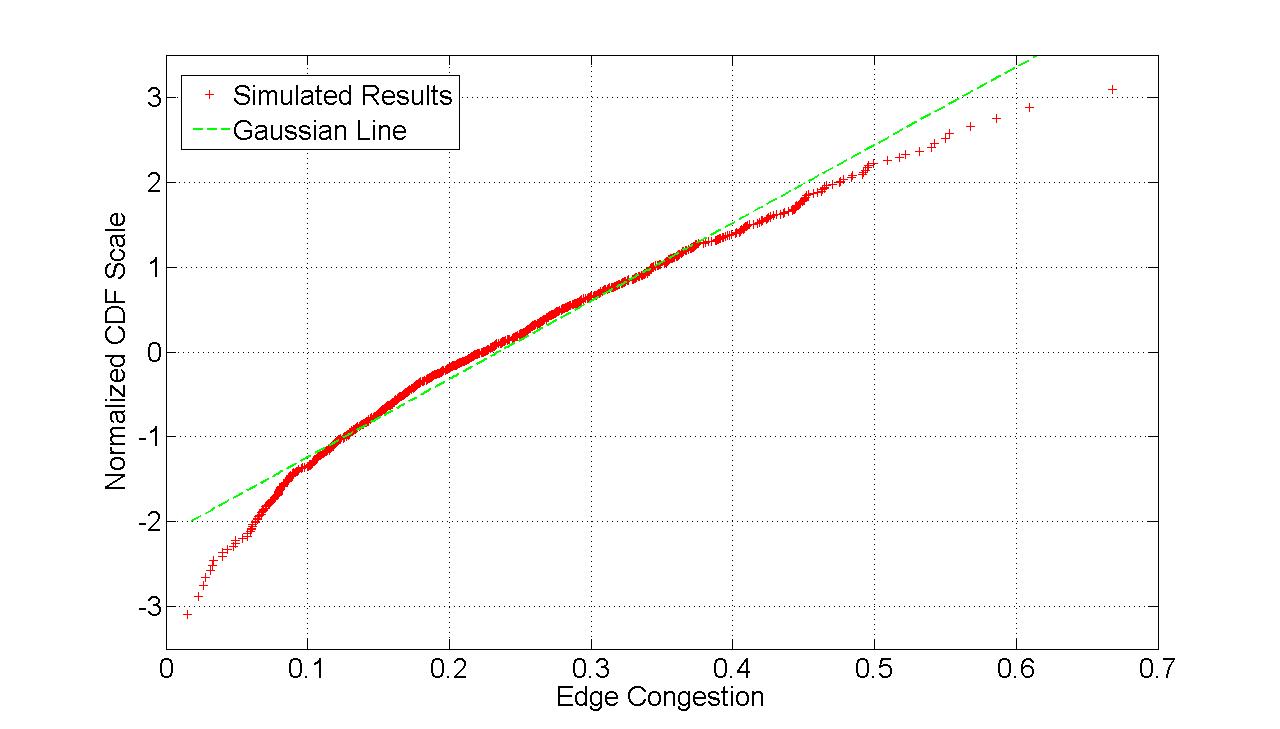}
        \label{fig:Abi_vol_e1_norm_test}} 

    \caption{Normal Probability Plots of the congestion on two edges in Abilene network (the T-Set is $\setA$) 
    \label{fig:Abi_vol_norm_test}}
\end{figure}

To check that the graphs indeed exhibit different behaviors, we
tested both plots for normality by using two different methods.
First, Figure \ref{fig:Abi_vol_norm_test} shows their corresponding
Normal Probability Plots (NPPs). NPPs are a graphical technique for
assessing whether or not a dataset is approximately normally
distributed~\cite{Chambers}. These NPPs were obtained by running
1,000 samples, each time measuring the induced edge congestion, and
then displaying the CDF of the sampled dataset, where the y-axis is
scaled so that a normally distributed dataset would yield an
approximate line. Note that the simulated results of the highly
loaded edge $e_{13}$ are very close to the Gaussian line, while
those of $e_{1}$ get rather far from it, which would seem to confirm
that the load on $e_{13}$ follows a near-Gaussian distribution while
the load on $e_{1}$ does not.

This is further confirmed by the Lilliefors test for goodness of fit
to a Gaussian distribution~\cite{key-Lillie}. The Lilliefors test
accepts the hypothesis of $e_{13}$ being a Gaussian (the test is
even significant at the stringent 1\% level). On the contrary, it
rejects the hypothesis of $e_{1}$ being a Gaussian (even at the
loose 20\% significance level). Therefore, both tests confirm that
one edge T-Plot displays a Gaussian behavior, while the other does
not.

The next Section explains the origin of the different behaviors of
the edge T-Plots of $e_{13}$ and $e_1$.

\section{To be or not to be (Gaussian)}
While it is clear that both T-Plots display different Gaussian
behavior, it is yet unclear why this is the case. In the remainder,
we will provide some limited intuition on the factors that might
make a distribution Gaussian or not --- but stress that we do not
provide any simple characterization, which is left for future
studies. For simplicity, we will also assume that we have
single-path routing.

\subsection{When flows are i.i.d.}
Assume for a moment that flows
$(i,j)$ are i.i.d. (independent and identically distributed). In
other words, when choosing a traffic matrix $D$ uniformly at random
from the T-Set, the distributions of $D_{ij}$ and $D_{kl}$ are
i.i.d. whenever the two flows $(i,j)$ and $(k,l)$ are distinct. This
might happen, for instance, if the T-Set is defined as the set of all
random matrices $D$ such that for all $i,j$, $0 \leq D_{ij} \leq
\frac{1}{n}$. Then, by the central limit theorem, when plotting the
T-Plots on links with an increasing number of flows crossing them,
the load distribution will become increasingly close to a Gaussian
distribution. This is because the load on a link $e$ equals
$\sum_{i,j} D_{ij} f_{ij}(e)$. Thus, if there are $k$ flows going
through a link, the load can be modeled as the sum of $k$ i.i.d.
random variables, which becomes increasingly Gaussian as $k$ grows,
with a difference bounded by Berry-Esseen--type
bounds~\cite{Esseen}. \\
Note that in our typical T-Sets $\setA$, $\setS$ and $\setP$, while
distinct flows are identically distributed, they are certainly not
independent, since there are additional constraints on the row and
column sums of the flow values.

Please recall that we assumed a single-path routing. In case of
multi-path routing, the mutual dependencies between the various
flows crossing each edge become more complex. It is hard to
determine whether switching from single-path routing to multi-path
routing would either increase or decrease the "entropy", i.e. the
rate of independence between different flows crossing the same edge.

\subsection{A counter-example}\label{counter-example}
Although the above explanation
provided some intuition on the Gaussian behavior, we will now show
that bad cases might still easily happen, even when $n$ grows to
infinity. Assume that in a given network, node $1$ sends traffic to
$p \cdot n$ destination nodes by using some unit-capacity edge $e$,
and does not use $e$ to send traffic to the other nodes, with $0 < p
< 1$ and $p \cdot n$ an integer number. Further assume that no other
flow uses $e$, and that the
T-Set is $\setP$, the set of permutation traffic matrices. \\

The matrix $f(e)$, which determines which (source, destination)
pairs route via $e$, contains $p \cdot n$ ones in the first row and
zeros elsewhere. For instance, if $p \cdot n = 3$ then $f(e)$ is:

\begin{equation}
f(e)=
\begin{pmatrix} 1&1&1&0&0&\dots&0\\
0&0&0&0&0&\dots&0\\
\hdotsfor[2]{4}\\
0&0&0&0&0&\dots&0\\
\end{pmatrix}
\label{Eq:f(e)}
\end{equation}

A fraction $p$ of the permutation matrices will generate a load of 1
on $e$, and a fraction $1-p$ a load of 0. The T-Plot PDF will be
equal to $1$ w.p. $p$, and $0$ otherwise, as shown in Figure
\ref{fig:Node alone} when $p=0.1$. The Gaussian model will be
generated by a Gaussian random variable of mean $p$ and variance
$p(1-p)$. As Figure ~\ref{fig:Node alone} shows, the Gaussian model
is clearly far from the real PDF, and normalizing to have the
modeled values at $0$ and $1$ add up to $1$ clearly does
not help.\\
Since this example is valid for all $n$, it teaches us that even when
the number of nodes and the number of flows passing through a link
go to infinity, the T-Plot can still {\em not} converge to a
Gaussian distribution. Therefore, much caution must be used. In this
case, it seems that this counter-example is possible because of the
high dependence between the flows that go through the edge $e$ - dependence
which is reflected by the fact that all the non-zero values in $f(e)$ belong
to the same row. Note that this dependence does not decrease with $n$.

\begin{figure}
\centering
\includegraphics
% [bb=0 920 1200 50,width=8cm, height=6cm]
[width = 0.66\columnwidth]
{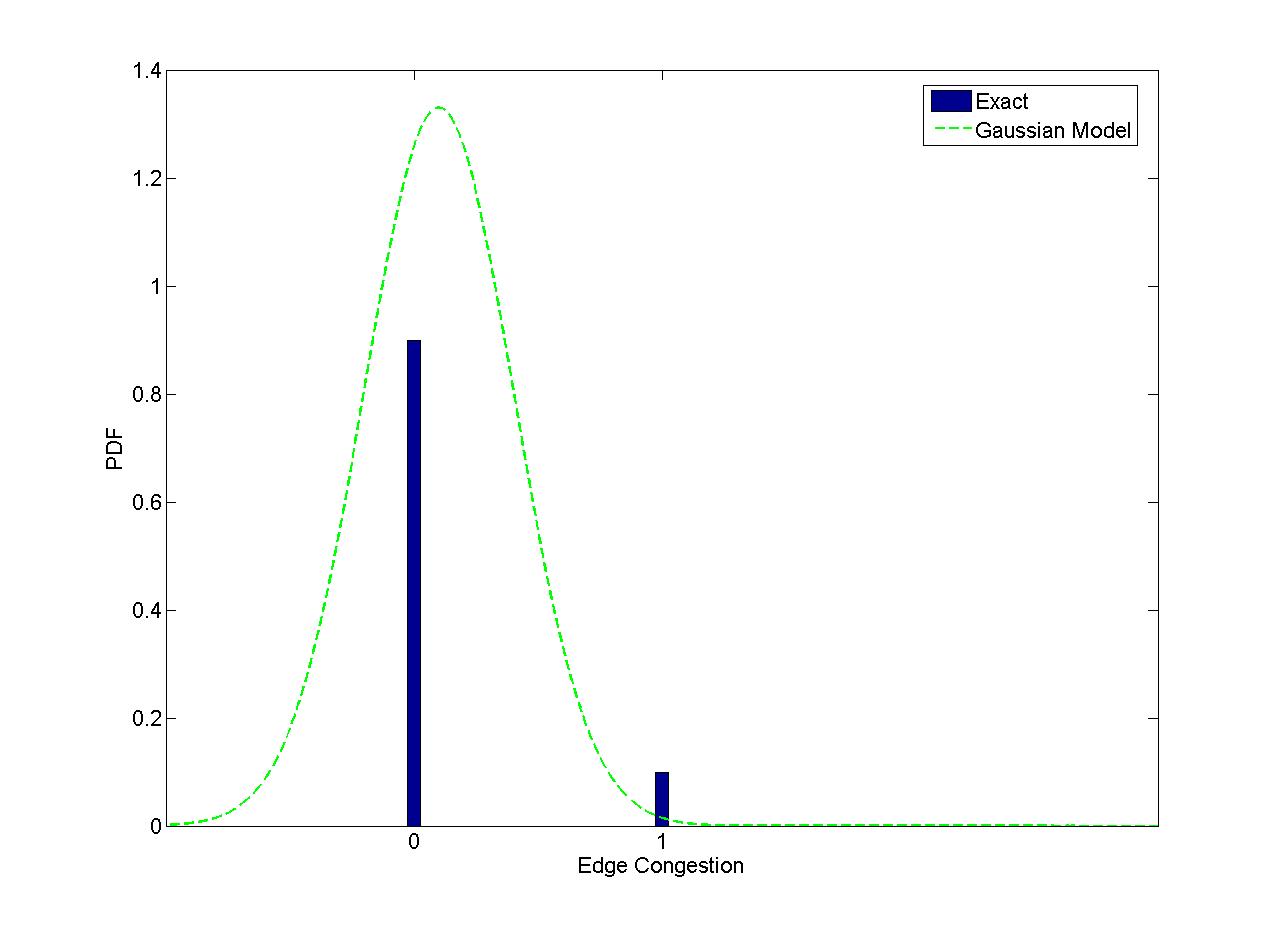}
\caption{PDF of the congestion for the scenario detailed in Sec. \ref{counter-example}
(the T-Set is $\setP$)} \label{fig:Node alone}
\end{figure}

\subsection {When the edge is overloaded}
The counter-example given in Section \ref{counter-example}
represents a rather "synthetic" scenario, where an edge is used only
for carrying some of the traffic of a single node. In real networks,
the overloaded edges tend to be located close to the "center", thus
carrying flows of many different nodes. As a consequence, the
distribution of the non-zero entries in $f(e)$ would be closer to
uniform, and the mutual dependencies between the various flows
crossing edge $e$ would decrease. This may give some intuition why
the T-Plot of edge $e_{13}$ (Figure ~\ref{fig:Abi_vol_e13_PDF_X}),
which is a "central" overloaded edge in Abilene network tends to be
much better modeled as Gaussian than the T-Plot of edge $e_1$
(Figure~\ref{fig:Abi_vol_e1_PDF}), which is a rather "peripheral"
under-loaded edge (see description of Abilene network architecture
in Chapter \ref{experimental_setup}).

\chapter{Congestion guarantees\label{sec:Congestion-guarantees}}\label{Sec6:guarantees}

\begin{figure}
    \centering 
    \subfigure[Congestion PDF] {
    \includegraphics
    % [bb=0 920 1200 50,width=5cm, height=5cm]
    [width = 0.66\columnwidth]
    {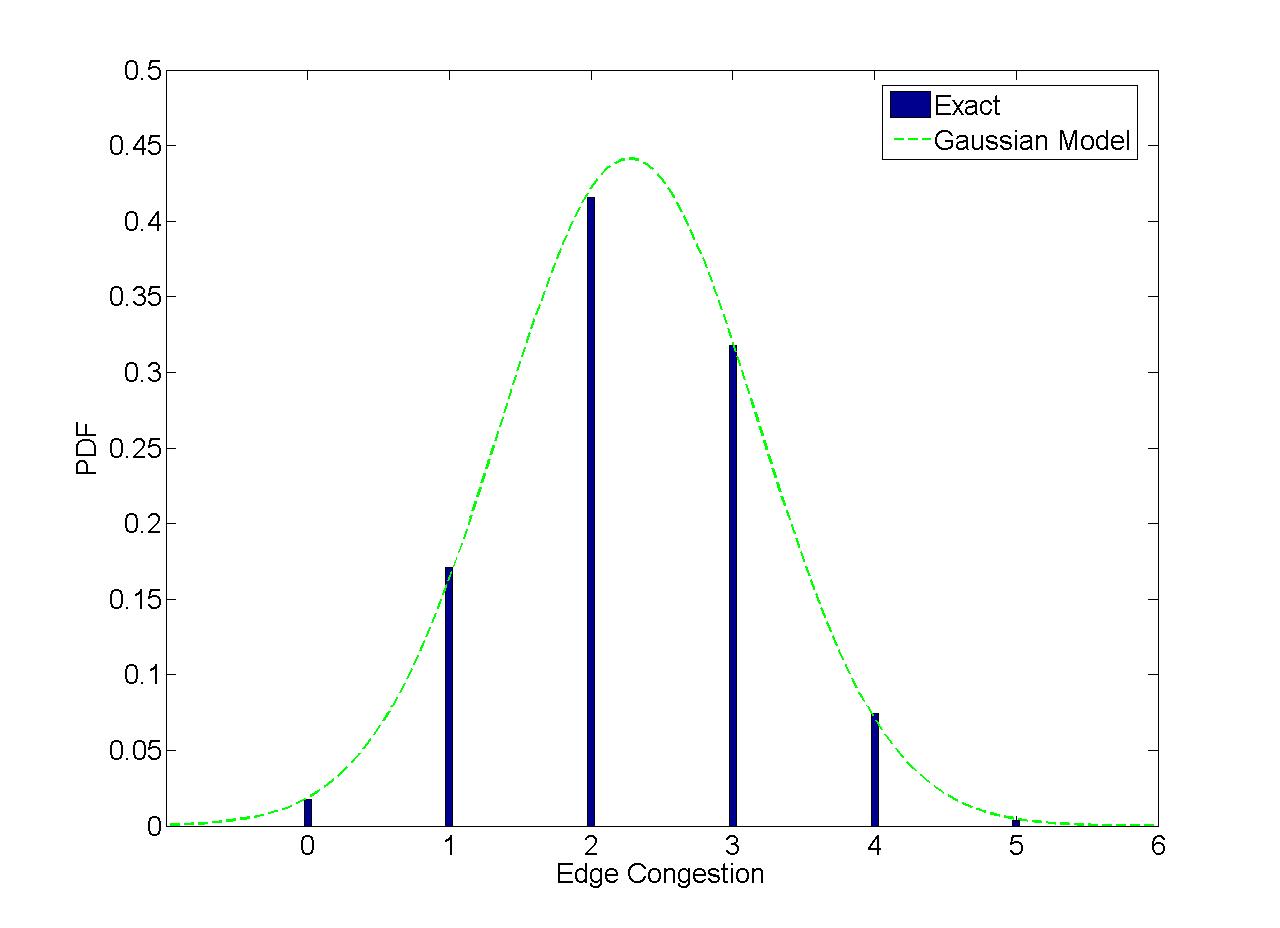}
    \label{fig:Abi_perms_e13_PDF}} 
    
    \subfigure[Congestion CDF]{
    \includegraphics
    % [bb=0 920 1200 50,width=5cm, height=5cm]
    [width = 0.66\columnwidth]
    {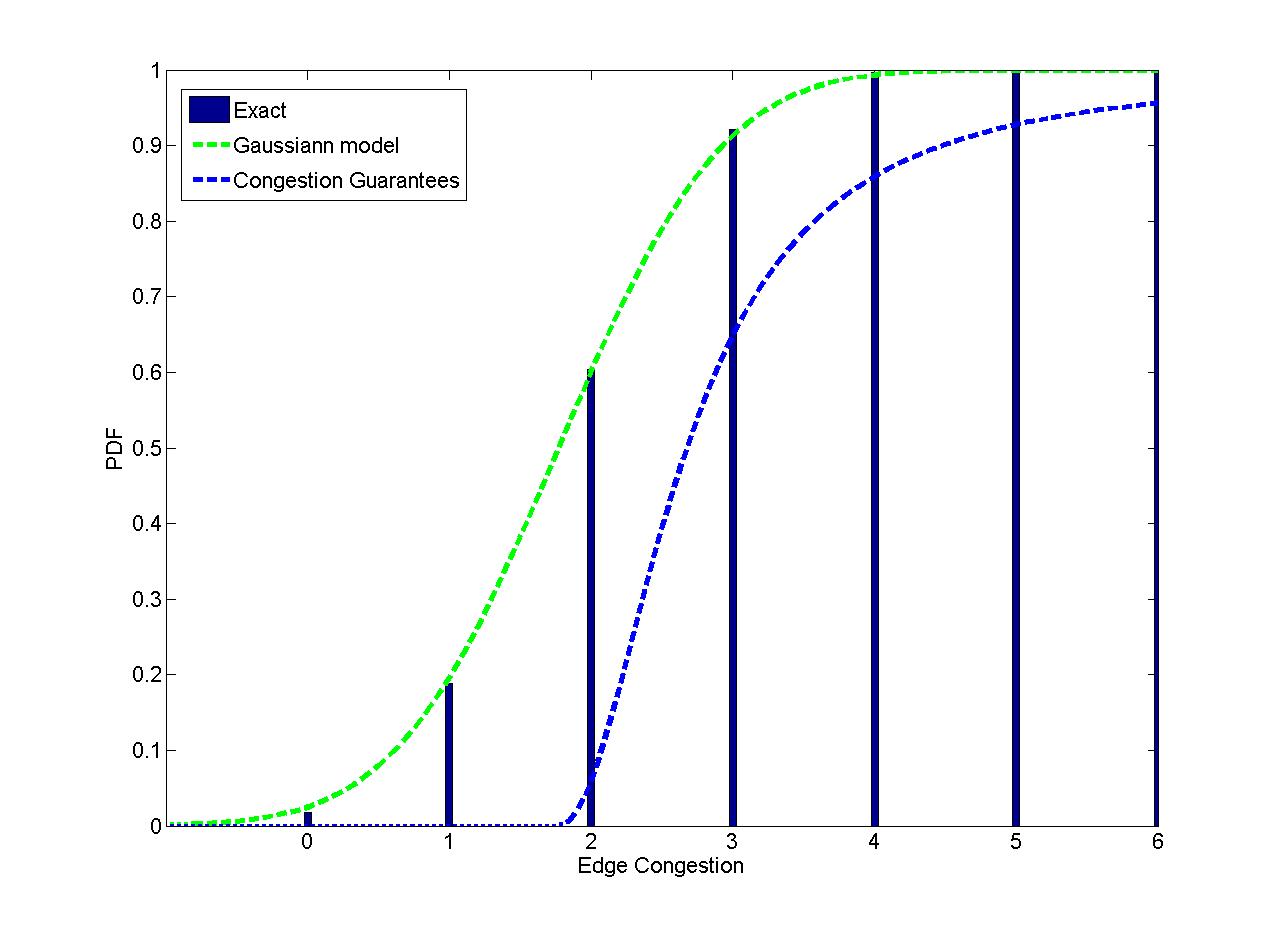}
    \label{fig:Abi_perms_e13_CDF}} 
    \label{fig_sim}
    \caption{Three views of the same
    T-Plot (edge 13, the T-Set is $\setP$)}
\end{figure}

Since it is not guaranteed that Gaussian approximations will be
correct, we are interested in providing performance bounds that are
guaranteed independently of the shape of the T-Plot. Further, we
would like to do so when only using the Gaussian parameters (mean
and variance). We will now show that it is indeed possible to use
these parameters to provide a bound on the probability that a given
edge would achieve less than 100\% throughput.

 Once we computed the average and the variance of the edge congestion, it is possible to bind the probability that a given edge would achieve less than 100\% throughput.
This technique is useful both for dynamically
identifying of bottlenecks \cite{key-PATHNECK}, and for online
routing algorithms, which try to minimize the interference between
possible future demands, such as MIRA \cite{key-MIRA}.

\textbf{Performance bounds --} Let us denote the average edge
congestion by $\mu$, and its standard deviation by $\sigma$. Let
\emph{X} denote the congestion imposed on a given edge \emph{e} by a
traffic matrix \emph{D} generated u.a.r. on the T-Set. Then, by
Chebyshev's one-tailed inequality with $k\geq 0$,

\emph{\begin{equation} Pr(X\geq
k\sigma+\mu)\leq\frac{1}{^{1+k^{2}}}\label{eq:Cheby 1
tail}\end{equation} }By definition, \emph{e} is saturated iff $X\geq
c(e)$. Therefore, the probability for \emph{e} to be saturated is
upper-bounded:

\begin{equation}
Pr(X\geq
c(e))\leq\frac{1}{^{1+\left[\frac{c(e)-\mu}{\sigma}\right]^{2}}}\label{eq:Cheby
1 tailed modified}\end{equation}
 Alternatively, given a desired certainty level \emph{G}, it is possible
to calculate a capacity \emph{c'(e)} that guarantees that at least a
fraction \emph{G} of the allowable traffic matrices would be served
without saturating \emph{e}. Transforming Equation (\ref{eq:Cheby 1
tailed modified}), we get:

\begin{equation}
c'(e)=\mu+\sigma\sqrt{\frac{G}{1-G}}\label{eq:single edge
CO}\end{equation}

\begin{figure}
    \centering 
    \subfigure[Throughput CCDF] {
    \includegraphics
    % [bb=0 920 1200 50,width=5cm, height=5cm]
    [width = 0.66\columnwidth]
    {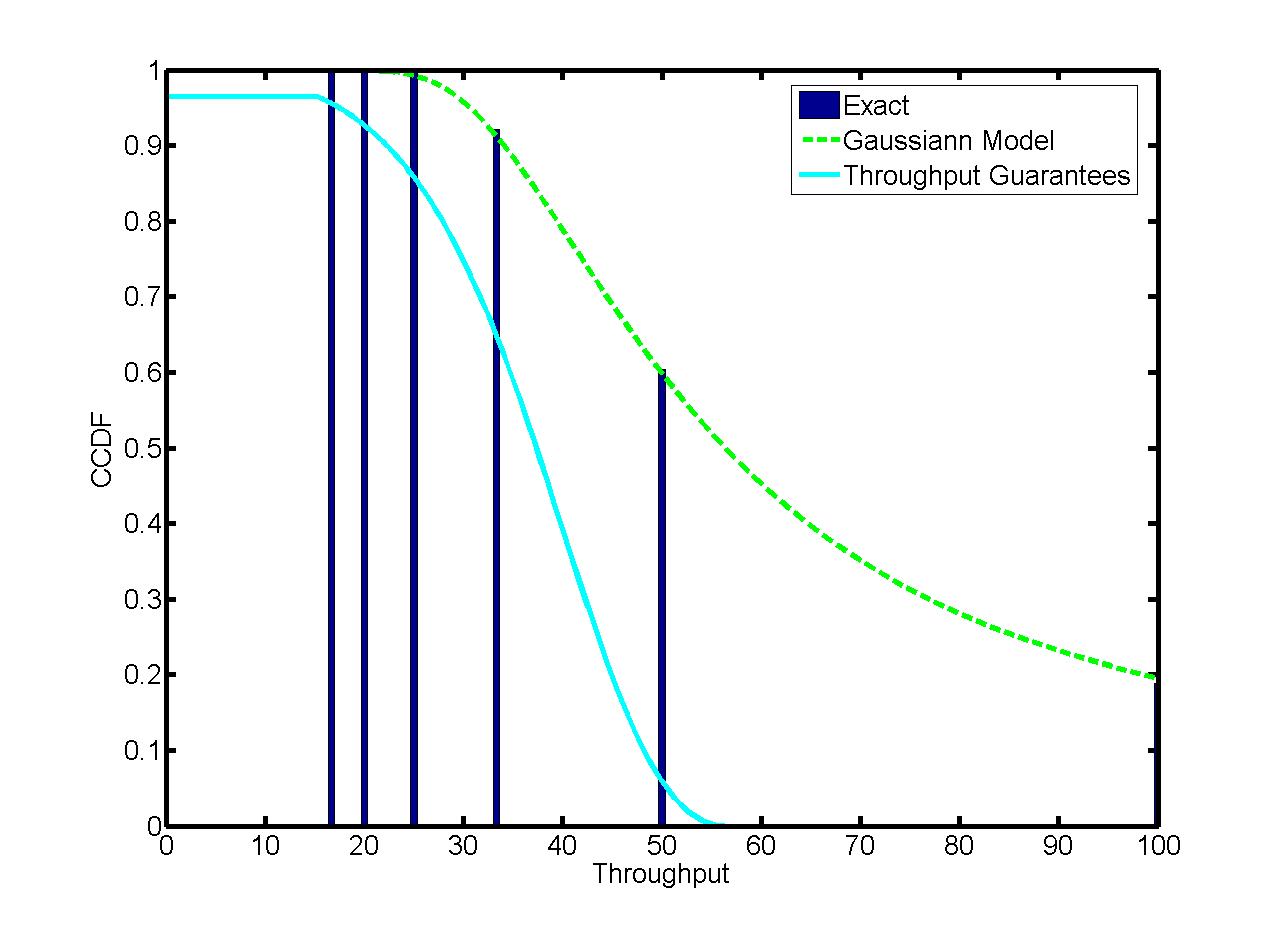}
    \label{fig:Abi_e13_TP_CCDF}} 
    \caption{Three views of the same
    T-Plot (edge 13, the T-Set is $\setP$, T-Plot of the throughput).}
\end{figure}

\textbf{Example --} Figure~\ref{fig_sim} provides an example of such
a bound. The T-Set is now assumed to be $\setP$, the set of
permutation traffic matrices, and the congestion is measured over
highly loaded edge $e_{13}$. Note that Figure~\ref{fig:Abi_perms_e13_CDF} provides
three different views of {\em the same} T-Plot, thus illustrating
how a network operator might choose different T-Plots depending on
the most significant parameters. \\
Figure~\ref{fig:Abi_perms_e13_PDF} displays the discrete PDF of the
edge congestion together with the Gaussian approximation, which is
again extremely close to the exact values (for once, these values
are exact, because all the $n!$=$11!$
permutation matrices were used in the sample). \\
The same T-Plot is also shown in a CDF form in
Figure~\ref{fig:Abi_perms_e13_CDF}, together with the new
performance bound (which didn't make sense to display on a PDF). For
each performance guarantee level, the performance bound now
guarantees a
sufficient corresponding amount of over-capacity. \\
Finally, the same T-Plot is also shown as the CCDF (Complementary
CDF) of the throughput in Figure~\ref{fig:Abi_e13_TP_CCDF}. For
instance, while some 92\% of the matrices get a guaranteed
throughput of at least $1/3$, the performance bound can only
guarantee this for 65\% of the matrices.

\chapter{Approximation and bounds of the global congestion
CDF}\label{Sec7:Approximation-and-bounds}

So far, we have mainly dealt with \emph{edge} congestions. We will
now deal with \emph{global} congestions. Of course, succeeding to
well approximate the \emph{global} T-Plots would mean obtaining a
performance model for the whole network. We will first provide a
simple approximation assuming independence, and then an upper-bound.

First assuming that all edge congestions are independent, i.e.
traffic matrices cause congestion at different links in an
independent manner, provides the following approximation:

\begin{equation}
GC_{CDF}(f,L)\approx\prod_{e\in E}EC_{CDF}(e,f,L)\label{eq:MEC_CDF
eq prod (single edges CDF)}\end{equation}

This is unfortunately quite often a poor approximation, and rather
plays the role of intuitive lower-bound, since matrices often cause
loads in a clearly dependent way. We thus look for an upper bound. A
trivial upper bound is obtained by the CDF of the most loaded edge
in the network:
\begin{equation}
GC_{CDF}(f,L)\leq\min_{e\in E}\{ EC_{CDF}(e,f,L)\}.\label{eq:GC leq
single_Edge_CDF}\end{equation} Further, if $e_{1}$ and $e_{2}$ are
two edges (e.g. the two most loaded edges in the network, which
could be the two different directions of the same link), then:
%the probability that the global
%congestion (the maximum on all edge congestions) is at least $x$ is
%lower-bounded by the probability that the maximum on those two edge
%congestions
\begin{eqnarray}
Pr\{ GC>x\}  & \geq & Pr\{ EC(e_{1})>x\textrm{ }\vee\textrm{ }EC(e_{2})>x\}\nonumber \\
&=       & Pr\{ EC(e_{1})>x\}+Pr\{ EC(e_{2})>x\}  \nonumber \\
&-       & Pr\{ EC(e_{1})>\textrm{ }x\textrm{ }\wedge EC(e_{2})>x\}  \nonumber\\
&\geq    &Pr\{ EC(e_{1})>x)+Pr\{ EC(e_{2})>x\} \nonumber \\
&-       & Pr\{ EC(e_{1})+EC(e_{2})>2x\},\label{eq:e1+e2 - e1_2}
\end{eqnarray}
where $Pr\{ EC(e_{1})+EC(e_{2})>2x\}$ is equal to
$1-EC_{CDF}(e_{1\_2},2x)$, using a dummy edge $e_{1\_2}$ for which
$f(e_{1\_2})=f(e_{1})+f(e_{2})$. Using $e_{1\_2}$, a similar upper
bound is obtainable as follows:
\begin{eqnarray}
Pr\{ GC \leq  x\}  & \leq & Pr\{ EC(e_{1})\leq x)\textrm{ }\wedge\textrm{ }EC(e_{2})\leq x)\}\nonumber \\
&\leq       &
EC_{CDF}(e_{1\_2},2x)\label{eq:upper_bnd}\end{eqnarray}
%\begin{eqnarray}
%GC_{CDF}(x) & = & Pr\{ GC(f)\leq x\}\nonumber \\
% & \leq & Pr\{ EC(e_{1},f)\leq x)\textrm{ }\wedge\textrm{ }EC(e_{2},f)\leq x)\}\nonumber \\
% & \leq & Pr\{ EC(e_{1},f)+EC(e_{2},f)\leq2x)\}\nonumber \\
% & = & EC_{CDF}(e_{1\_2},f,2x)\label{eq:upper_bnd}\end{eqnarray}
A stricter upper bound may then be defined as the minimum of the
three bounds (\ref{eq:GC leq single_Edge_CDF}), (\ref{eq:e1+e2 -
e1_2}) and (\ref{eq:upper_bnd}).

\chapter{Capacity allocation scheme}\label{Sec8:Improving-the-capacity}

Let's denote by $C_i$ the capacity of edge $i$, and by $\mu_i$,
$\sigma_i$ the mean and standard deviation of the load on edge $i$,
respectively. We now suggest to allocates to edge $i$ a capacity of
$c(i) = \mu_i + k\sigma_i$, where we use the same value of \emph{k}
for all edges. The total capacity required, as a function of
\emph{k}, is:
\begin{equation}
\sum^{|E|}_{i=1}c(i)= \sum^{|E|}_{i=1}\mu_i +
k\sum^{|E|}_{i=1}\sigma_i \label{eq:Cap allocation - sec 5}
\end{equation}
Note that when the total capacity is constrained to be smaller than
the sum of the average-case edge congestions, \emph{k} is negative.

We will now show that this capacity allocation minimizes the
probability that the network is saturated when using two
approximation assumptions: first, the loads on different edges are
independent; and second, the edge T-Plots obey a Gaussian model with
the same standard-deviation.

We remind that a network is saturated if at least one edge in it is
saturated; an edge $e$ is saturated if the edge congestion on it is
at least 1, i.e. the flow crossing $e$ is not below its capacity.

\begin{mytheorem} \label{thm: Thm 3} Assume that the T-Plots of all edges $i$ are independent and Gaussian of mean $\mu_i$ and same standard-deviation $\sigma$. Then allocating to each edge $i$ a
capacity $C_i = \mu_i + k\sigma$, where $k$ is a real constant,
minimizes the probability that the network is saturated.
\end{mytheorem}
\begin{proof}
We denote the total available capacity by $C$ , i.e., a legal
capacity allocation satisfies $\sum_{i}C_i \leq C$. We would like to
maximize the probability that the global congestion is below 1.
Formally, we want to find
\begin{equation*}
\argmax_{\{C_1, \cdots, C_{|E|} \} }   Pr\{GC(f,C_1,\cdots,,C_{|E|})
< 1\}
\end{equation*}
\begin{equation*}
\mathrm{s.t. } \sum^{|E|}_{i=1}{C_i} =C
\end{equation*}
Using the independence assumption, we want to maximize the product
of the probabilities of edge congestions on all edges:
\begin{equation}
Pr\{GC(f,C_1,\cdots,C_{|E|}) < 1\} = \prod^{|E|}_{i=1} {EC_{CDF}(e_i
(C_i), f, 1)}  \label{eq:product}
\end{equation}
Let $H(x)$ be the standard Gaussian CDF, and $h(x)$ the PDF (its
derivative). Then,
\begin{equation}
EC_{CDF}(e_i (C_i),f, 1))=H((C_i-\mu_i)/\sigma).
\end{equation}
Using Lagrange multipliers and differentiating
Equation~(\ref{eq:product}), we find that the objective function is
maximized when for each $i \neq j,$
\begin{equation}
\frac{EC_{PDF}(e_i (C_i),f, 1))}{EC_{CDF}(e_i (C_i),f, 1))} =
\frac{EC_{PDF}(e_j (C_j),f, 1))}{EC_{CDF}(e_j (C_j),f,
1))},\end{equation} i.e.
\begin{equation}
\frac{h((C_i-\mu_i)/\sigma)}{H((C_i-\mu_i)/\sigma)} =
\frac{h((C_j-\mu_j)/\sigma)}{H((C_j-\mu_j)/\sigma)}.
 \label{Eq:CDF_div_PDF}
\end{equation}
Equation (\ref{Eq:CDF_div_PDF}) is solved when $C_i = \mu_i
+k\sigma$ for each $i$, where $k$ is {\em the same} constant for
each $i$, because all the PDF probabilities are then equalized (the
standard-deviations are equal) as well as the CDF probabilities.
Therefore, this solution maximizes the objective function. Note that
for different standard deviations, this Lagrangian-based method can
also directly bring the optimal, yet less elegant, capacity
allocation solution, by solving Equation~(\ref{Eq:CDF_div_PDF})
using methods of numerical analysis.
\end{proof}

\chapter{Simulations}\label{Sec9:Simulations}

\begin{figure}
    \centering
    \includegraphics
    % [bb=0 920 1200 50,width=6 cm, height=5cm]
    [width = 0.66\columnwidth]
    {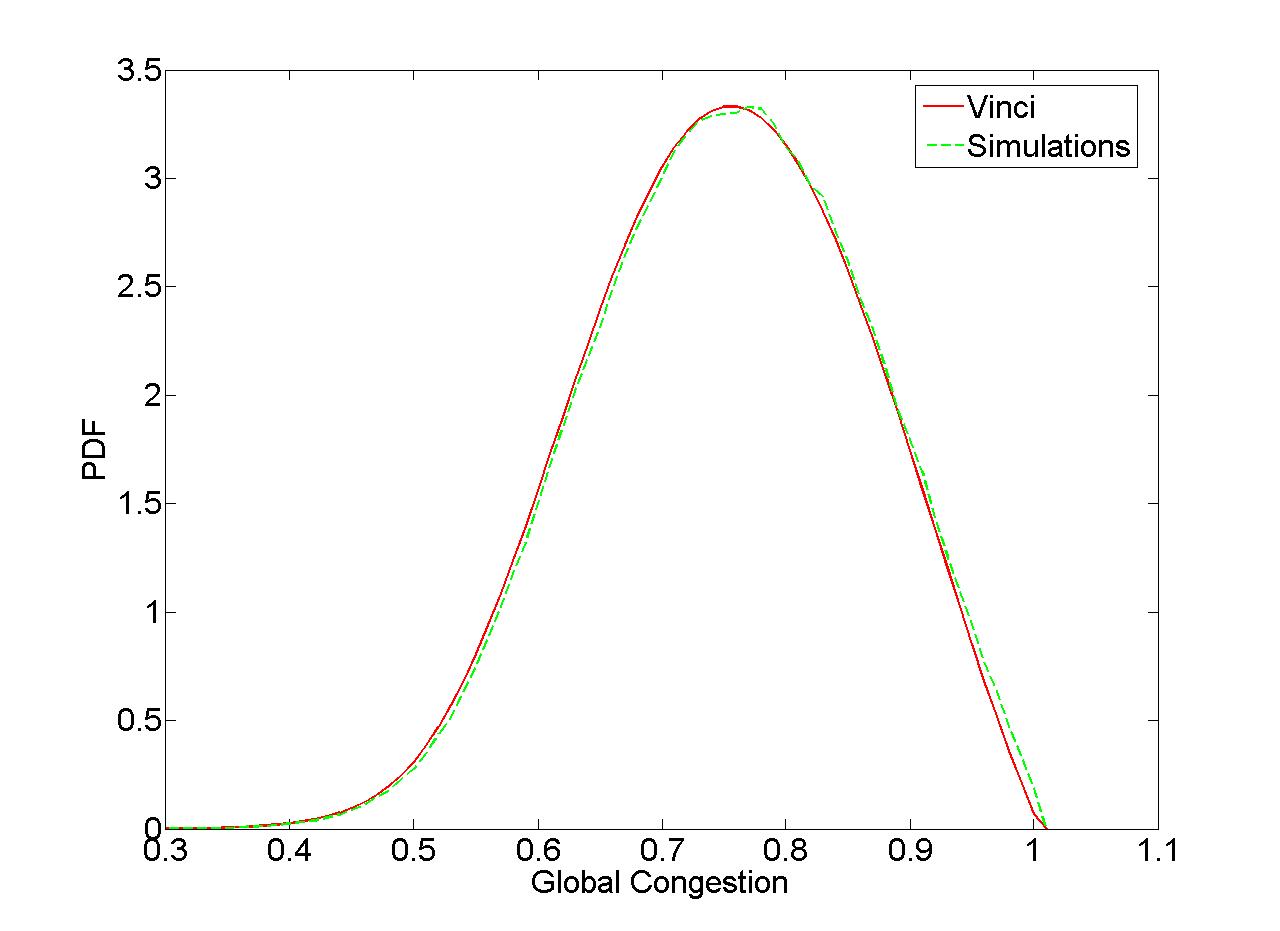}
    \caption{PDF of the global congestion in the $2 \times 2$ cube (when
    the T-Set is $\setA$)} \label{fig:Vinci}
\end{figure}

\section{T-Set representation}\label{sub:T-Set-representation}
We need the ability to pick a traffic matrix u.a.r. from the T-Set
for two reasons. First, to obtain a close approximation of
the average and the variance of the edge congestion (as explained in
Section ~\ref{AppendixC}). Second, to compare the models
developed in the previous chapters to simulation results. While it is
easy to generate a random permutation, picking a traffic matrix
u.a.r. from $\setS$ or $\setA$ is harder. To do so, we used
random-walk sampling. After starting from some matrix in the T-Set
(e.g., the null matrix), we computed a new matrix each time. To do
so, at each step, we compute the matrix change $\Delta$, and add it
to the current matrix $D$, so that the new matrix is $D'=D+\Delta$.
If $D'$ is in the T-Set, we move to it, and reject it otherwise.
When the T-Set was $\setA$, we performed a Gaussian version of the
ball walk \cite{key-Random_walk}, i.e. the next point in the walk is
normally distributed around the current point such that
$\Delta[i,j]\sim N(0,1/2n)$ for each $i,j$ (of course, other
random-walk approaches can be used if the guaranteed rate of
convergence is important to the user, as explained further in
~\cite{key-Random_walk}). When the T-Set was $\setS$, we picked at
each step $4$ integers $(i_{1},j_{1},i_{2},j_{2})$ u.a.r. from
$[1,n]$, generated a normally distributed step $\delta\sim
N(0,1/2n)$, and assigned
$\Delta[i_{1},j{}_{1}]=\Delta[i_{2},j{}_{2}]=-\delta$ and
$\Delta[i_{1},j{}_{2}]=\Delta[i_{2},j{}_{1}]=\delta$. The rest of
the entries in $\Delta$ are zeros. Note that this generation process
requires checking at each step whether the next matrix is indeed in
the required T-Sets. In any case, speed was not a problem: for
instance, it takes a few minutes to generate a million random
samples in the Abilene case ($n=11$) using an unoptimized Matlab
script. However, we needed to check that the sample distribution
indeed converges, in the expected speed, to the uniform distribution
on the T-Set. To do so, we used three different tests.

First, we used the fact that computing the T-Plots when the T-set is
$\setA$ may also be expressed as a volume computation
problem~\cite{key-Typical_vs_WC}. For each load value $L$, the
region
$$\setR=\{D\in \setA | GC(f,D) < L\}$$ forms a convex polytope. The
ratio between the volume of $\setR$ and the volume of the polytope
$\setA$ is equal to $GC_{CDF}^{T}(f,L)$. (This method can be
extended to $\setS$, as it is possible to compute the volume of
$\setS$ with $n$ nodes by computing the volume of $\setA$ with $n-1$
nodes \cite{key-vol_of_poly_of_SD}.) Therefore, we could compare the
results of our Monte Carlo simulations using 1 million
u.a.r.-generated matrices to the exact results obtained by Vinci
\cite{key-Vinci}, a standard convex-polytope volume computation
tool. As noted before, computation of the volume of polytopes is
\emph{\#P-C} \cite{key-Vol_is_P-C}, and therefore we used a small
sample network: the $2 \times 2$ cube with $n=4$ and DOR
routing~\cite{key-DOR}. Figure \ref{fig:Vinci} is a T-Plot of the
global congestion at this network.  The graph shows that the
simulation results are extremely close to Vinci's exact results.

The comparison with Vinci strongly validates the approach with a small
$n$. However, as $n$ becomes larger, there is no way to verify the results. We checked two other, much weaker
properties. First, we checked that the resulting T-Plot is
independent of the initial matrix. Second, after $m$ experiments of
u.a.r. generated points, we verified that the variance did indeed
decrease as $O(1/m)$ and obey the formula $var[X]=\frac{p(1-p)}{m}$,
where $X$ is the indicator r.v. representing the fact that the
sample is in some bin and $p=EX$.\\

\begin{figure}
    \centering
    \includegraphics
    % [bb=40 230 570 530,width=9cm, height=6cm]
    [width = 0.66\columnwidth]
    {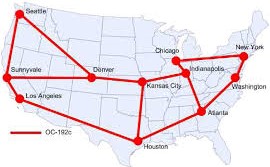}
     \caption{Abilene backbone
    network} \label{fig:Abi_dist_map}
\end{figure}

\section{Experimental setup}\label{experimental_setup}

%% Abi topo: STTL  = 1; SNVA = 2; DNVR = 3; LOAS = 4; KSCY = 5; HSTN = 6; CHIN = 7; IPLS = 8; NYCM = 9; WASH = 10; ATLA = 11;
% e1: (1,2)

We use Abilene as a sample backbone network \cite{key-Abi}, and
provide a simple model of its characteristics. We assume a shortest
path routing according to the metric shown in Figure \ref{fig:Abi_dist_map}.

All link capacities are 10Gbps \cite{key-Abi_is_10Gbps}. We further
assume that each node can send and receive up to 10Gbps (homogeneous
case). Obviously, we may normalize it by using units of 10Gbps; that
is, assume that nodes have a unit injection and ejection bandwidth,
and that the edge capacities are all equal to 1. We will later
revisit this assumption, which is rather pessimistic in the sense
that the aggregated rate received or initiated by each of Abilene's
nodes is usually far below 10Gbps \cite{key-Abi-traffic}. In this
Thesis, we focus on the most loaded edge, called $e_{13}$, which
corresponds to the (Kansas City, Indianapolis) directed edge; and on
the least loaded edge, called $e_1$, which corresponds to the
(Seattle, Sunnyvale) directed edge.

Abilene contains 11 nodes, so the number of permutations is $11!$,
and an exhaustive examination is feasible. Thus, the T-Plots for the
set of permutations, $\setP$, and for the set of derangements,
$\setP_{d}$, include the exact results. However, when the network
contains more nodes, an exhaustive examination is not feasible anymore,
and the common method is to use Monte Carlo simulations
\cite{key-WC_traffic,key-O1TURN,key-TP_centric,key-GOAL}. Indeed,
neither are exact results obtainable when the T-Set is continuous.
For such cases, we simulated 1 million u.a.r.-generated traffic
matrices, as explained above.

\section{Edge T-Plots}
Figures \ref{fig_e13_P_Pd} and \ref{fig_e13_DS_DSS} show the PDF of
the edge congestion on edge $e_{13}$. Figure
\ref{fig:Abi_perms_e13_PDF_again} is identical to Figure
~\ref{fig:Abi_perms_e13_PDF} which shows the PDF of the edge
congestion on edge $e_{13}$, where the T-Set is $\setP$, and is
given here again for convenience. Figure \ref{fig:Abi_Pd_e13_PDF}
shows the PDF of the edge congestion on edge $e_{13}$, where the
T-Set is $\setP_{d}$. The Gaussian approximation for this case is
again very accurate. A comparison between Figure
~\ref{fig:Abi_perms_e13_PDF_again} and Figure
\ref{fig:Abi_Pd_e13_PDF} shows that the probability of an edge to
become congested increases when switching from permutations to
derangements, since no efficient routing algorithm would waste
network resources to route a packet from a node to itself. This
shows how assuming that the traffic matrix belongs to the whole set
of permutations, which is often done in the literature
\cite{key-WC_traffic,key-O1TURN,key-TP_centric,key-GOAL}, is a bit
too optimistic. The theoretical computation shows that the ratio
between the mean load when the T-Set is $\setP_d$ and the mean load
when the T-Set is $\setP$ is $n/(n-1)$ (see Chapters
\ref{SubSec:avg-variance} and \ref{SubSec:avg-variance-more}).

\begin{figure}
    \centering 
    \subfigure[Congestion PDF on $e_{13}$ (over $\setP$)]{
    \includegraphics
    % [bb=0 920 1200 50,,width=7cm, height=5cm]
    [width = 0.66\columnwidth]
    {Abi_e13_PDF.jpg}
    \label{fig:Abi_perms_e13_PDF_again}}
    \subfigure[Congestion PDF on $e_{13}$ (over $\setP_{d}$)] {
    \includegraphics
    % [bb=0 920 1200 50,width=7cm, height=5cm]
    [width = 0.66\columnwidth]
    {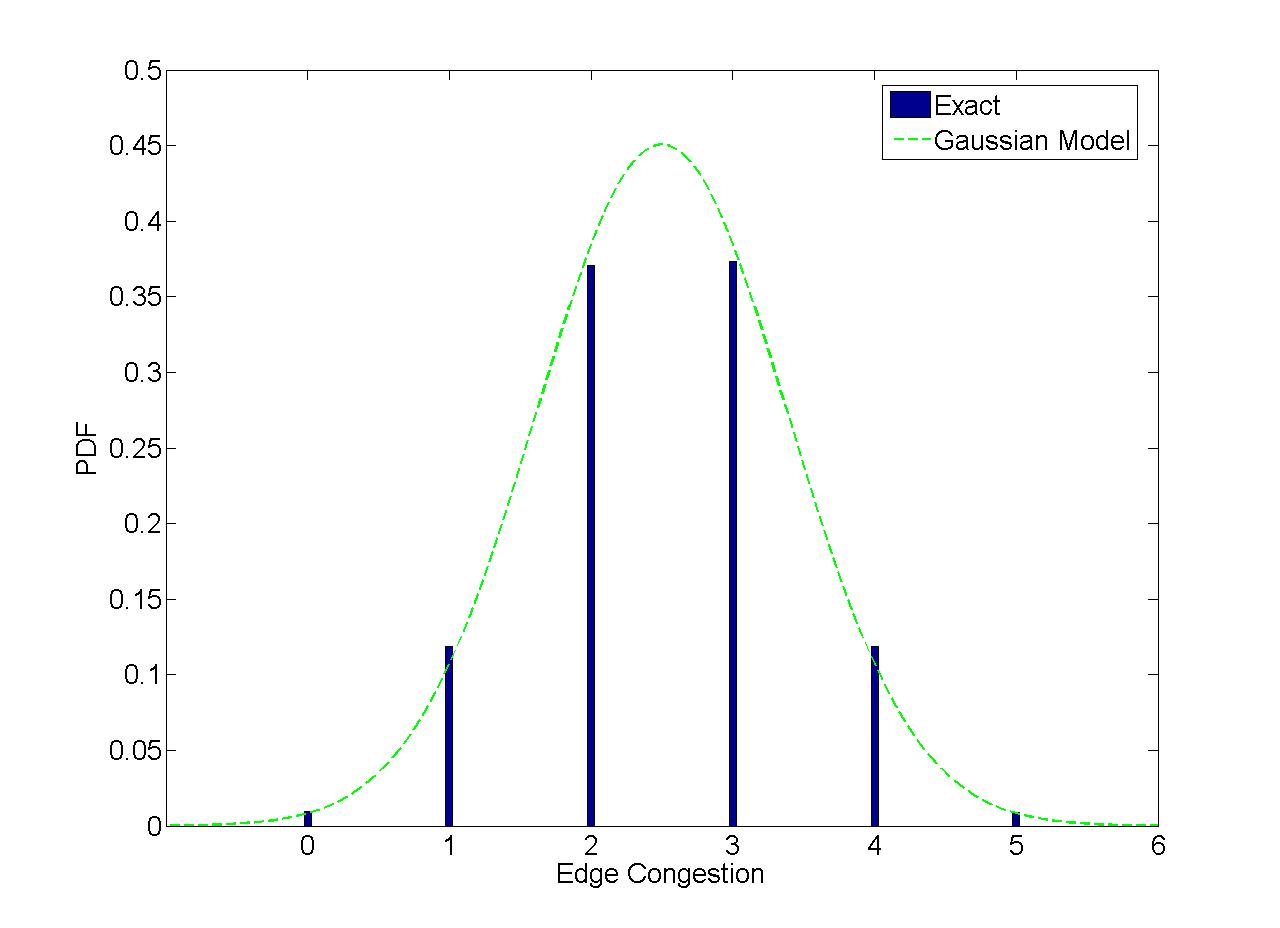}
    \label{fig:Abi_Pd_e13_PDF}} 
    \caption{Congestion PDF on $e_{13}$ over discrete T-Sets}
    \label{fig_e13_P_Pd}
\end{figure}
    
\begin{figure*}
    \centering
    \subfigure[Congestion PDF on $e_{13}$ (over $\setS$)]
    {
    \includegraphics
    % [bb=0 920 1200 50,width=7cm, height=5cm]
    [width = 0.66\columnwidth]
    {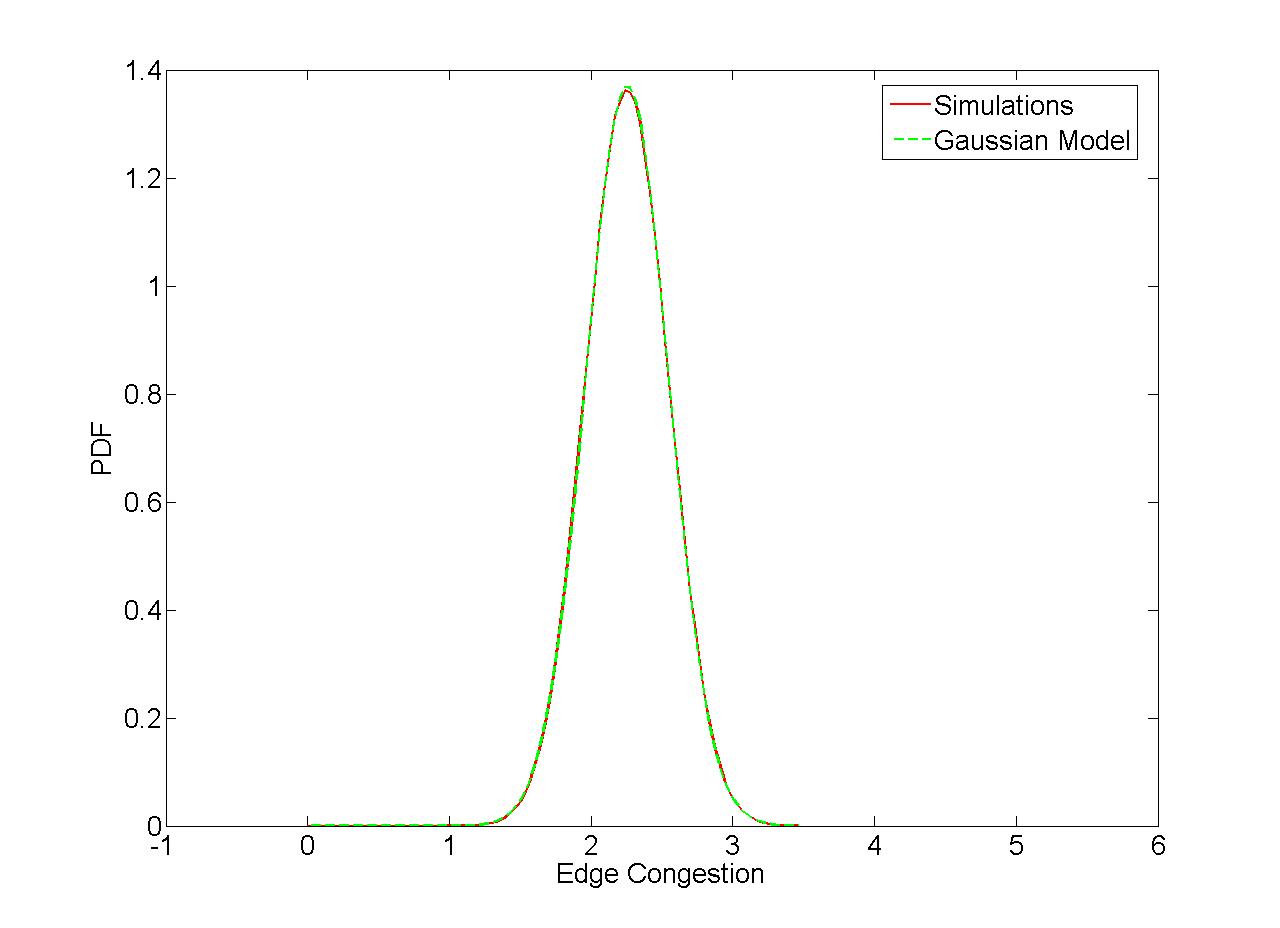}
    \label{fig:Abi_DS_e13_PDF}}
    \subfigure[Congestion PDF on $e_{13}$ (over $\setA$)]
    {
    \includegraphics
    % [bb=0 920 1200 50,width=7cm, height=5cm]
    [width = 0.66\columnwidth]
    {Abi_vol_e13_PDF.jpg}
    \label{fig:Abi_vol_e13_PDF_X}}  
    \caption {Congestion PDF on $e_{13}$ over continuous T-Sets}
    \label{fig_e13_DS_DSS}
\end{figure*}

Figure ~\ref{fig:Abi_DS_e13_PDF} shows the PDF of the edge
congestion on edge $e_{13}$, where the T-Set is $\setS$. The
Gaussian model fares well again, as in the other T-Plots involving
$e_{13}$. Figure \ref{fig:Abi_vol_e13_PDF_X} is identical to Figure
~\ref{fig:Abi_DSS_e13_PDF} which shows the PDF of the edge
congestion on edge $e_{13}$ where the T-Set is $\setA$, and is given
here again for convenience. A comparison between
Figure~\ref{fig:Abi_perms_e13_PDF_again} and Figure
\ref{fig:Abi_DS_e13_PDF} reveals that when switching from $\setP$ to
$\setS$, the variance of the congestion decreases, which is
intuitively explained by the fact that the permutations represent
the edge cases, which have a larger tendency to be either extremely
overloaded or extremely underloaded. Note that a further decrease in
the variance (and also in the average) of the congestion happens
when switching from $\setS$ (Fig. \ref{fig:Abi_DS_e13_PDF}) to
$\setA$ (Fig. \ref{fig:Abi_vol_e13_PDF_X}). This decrease happens because $\setA$ includes also traffic matrices whose total sum is
smaller than \emph{n}, thus typically representing lower traffic
demands, and smaller probability for edge cases.

Figures \ref{fig_e1_P_Pd} and \ref{fig_e1_DS_DSS} are similar to
Figures \ref{fig_e13_P_Pd} and \ref{fig_e13_DS_DSS} respectively,
but now the distributions are plotted for edge $e_1$. Figure
\ref{fig:Abi_vol_e1_PDF_again} is identical to Figure
~\ref{fig:Abi_vol_e1_PDF} which shows the PDF of the edge congestion
on edge $e_{1}$ where the T-Set is $\setA$, and is given here again
for convenience. It can be seen again that there is a continuous
decrease in the mean edge congestion when switching from $\setP_d$
to $\setP$, further to $\setS$ and finally to $\setA$. Note that for
none of the plots does the Gaussian Model fare well - see the
discussion in Chapter \ref{Sec5:gaussian}.

\begin{figure}
    \centering
    \subfigure[Congestion PDF on $e_{1}$ (over $\setP$)]
    {
    \includegraphics
    % [bb=0 920 1200 50,width=7cm, height=5cm]
    [width = 0.66\columnwidth]
    {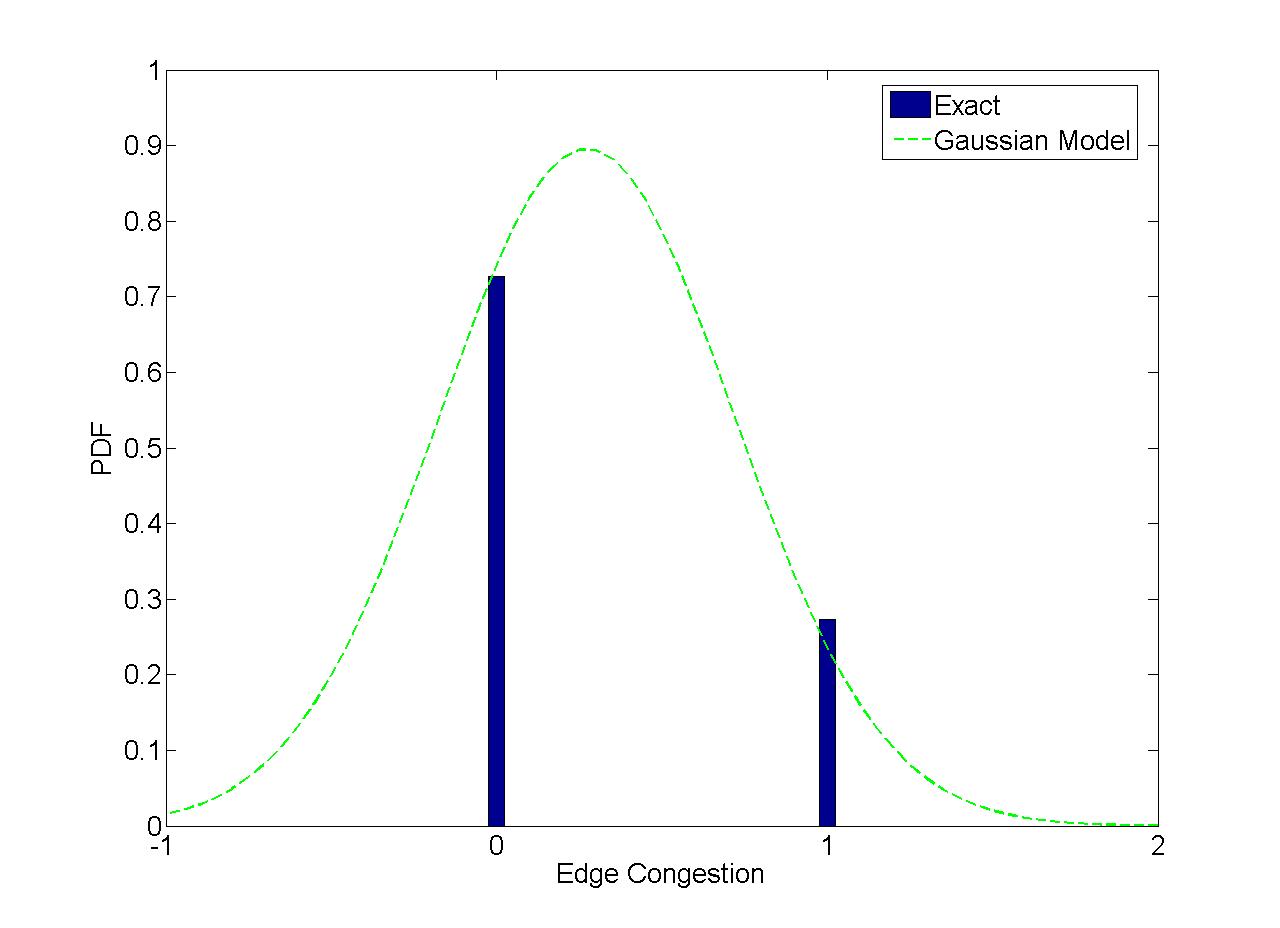}
    \label{fig:Abi_perms_e1_PDF}}
    \hfill \subfigure[Congestion PDF on $e_{1}$ (over $\setP_{d}$)]
    {
    \includegraphics
    % [bb=0 920 1200 50,width=7cm, height=5cm]
    [width = 0.66\columnwidth]
    {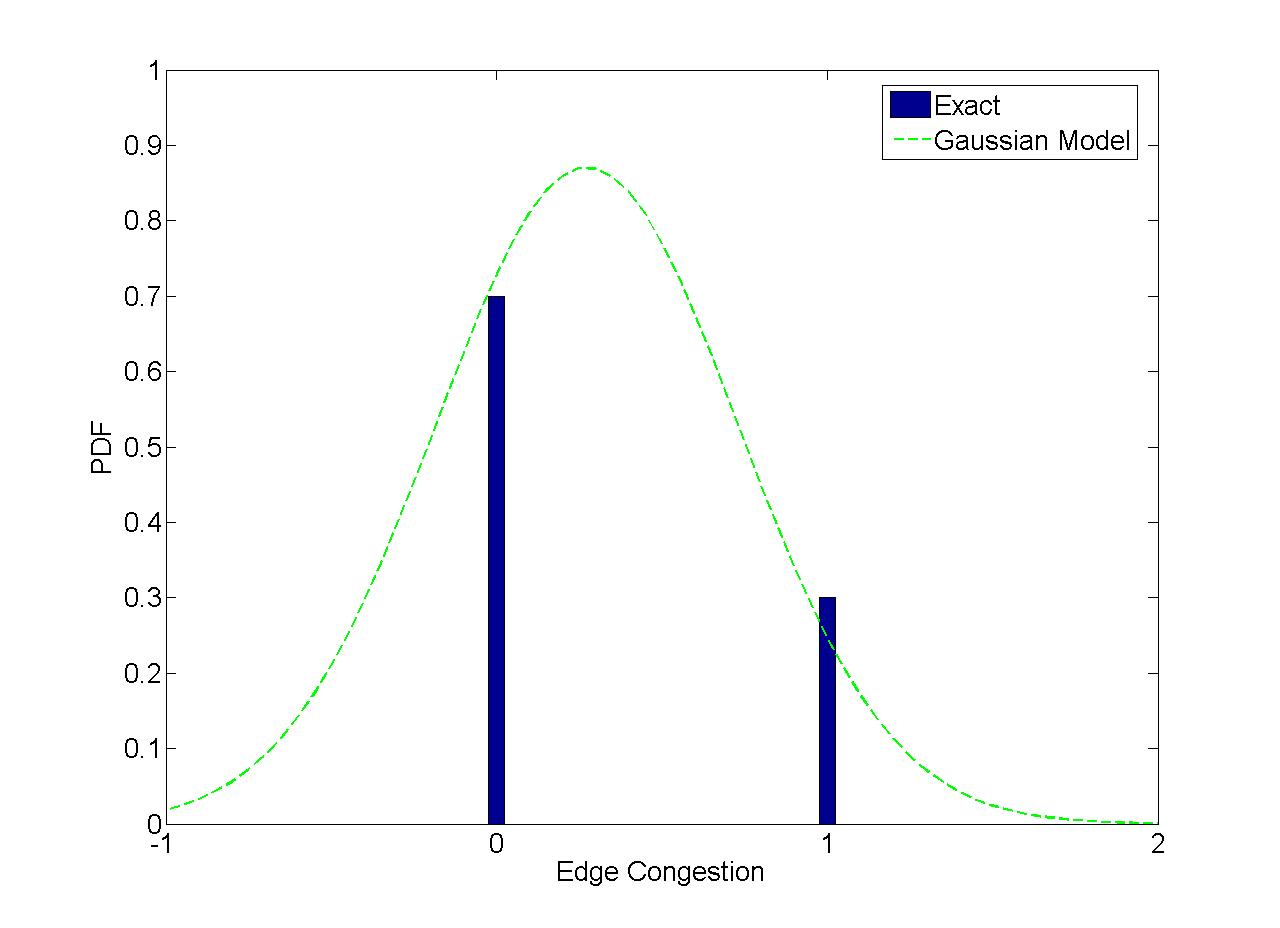}
    \label{fig:Abi_Pd_e1_PDF}} 
    \caption{Congestion PDF on $e_{1}$ over discrete T-Sets}
    \label{fig_e1_P_Pd}
\end{figure}

\begin{figure}
    \centering
    \subfigure[Congestion PDF on $e_{1}$ (over $\setS$)]{
    \includegraphics
    % [bb=0 920 1200 50,width=7cm, height=5cm]
    [width = 0.76\columnwidth]
    {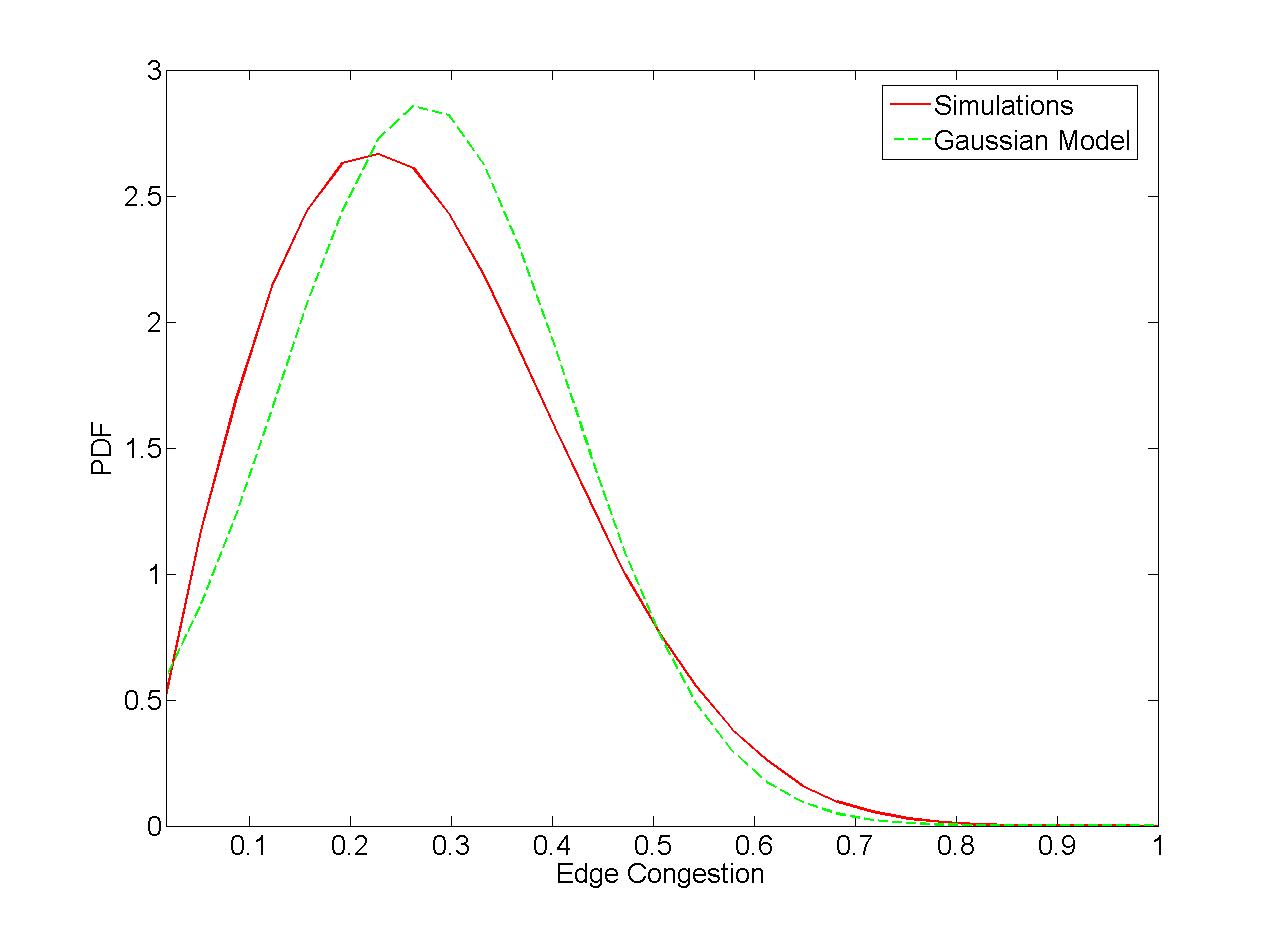}
    \label{fig:Abi_DS_e1_PDF}}
    \vspace{3 cm}
    \subfigure[Congestion PDF on $e_{1}$ (over $\setA$)]{
    \includegraphics
    % [bb=0 920 1200 50,width=7cm, height=5cm]
    [width = 0.76\columnwidth]
    {Abi_vol_e1_PDF.jpg}
    \label{fig:Abi_vol_e1_PDF_again}}  
    \caption {Congestion PDF on $e_{1}$ over continuous T-Sets}
    \label{fig_e1_DS_DSS}
\end{figure}

\section{Global T-Plots}
Figure~\ref{fig:Abi_vol_MEC_PDF} shows the PDF of the global
congestion in Abilene. It can be seen that it is well fitted by a
Gaussian distribution. Note that here, contrarily to all other
places, we only fitted the Gaussian distribution without using
\emph{a-priori} models. In other words, the mean and the variance of
the distribution, which are necessary for plotting the Gaussian, are
those found using simulations, and not by theoretical models -
as we do \emph{not} have such for \emph{global} T-Plots. Note that
the maximum of $n$ i.i.d. Gaussian random variables does {\em not}
behave as a Gaussian random variable (it follows a Gumbel
distribution~\cite{Gumbel}), and thus one must be careful with the
conclusions taken from this plot. It is reasonable to assume that
most of the Gaussian behavior comes from the few most loaded links,
like $e_{13}$.

\begin{figure*}
\centering{
{
\includegraphics
% [bb=0 920 1200 50,width=8cm, height=6cm]
[width = 0.76\columnwidth]
{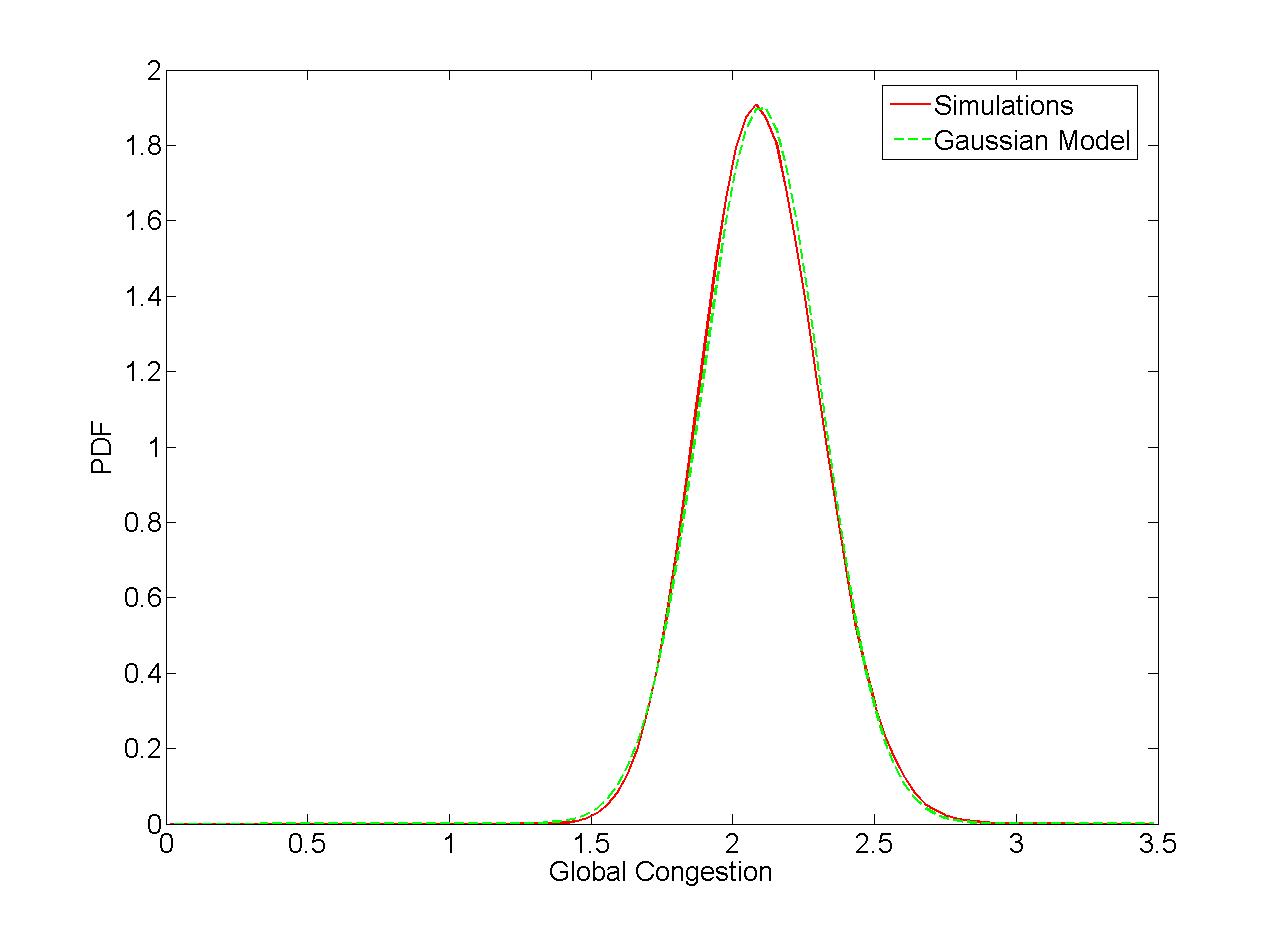}
\caption{Global congestion PDF: over $\setA$}
\label{fig:Abi_vol_MEC_PDF}} }
\end{figure*}

Figure \ref{fig:Abi_vol_MEC_CDF} shows the CDF of the global
congestion in Abilene. We calculate the approximation and the upper bounds using the models developed in Chapter
~\ref{Sec7:Approximation-and-bounds}. Note that switching from the
Gaussian approximation to the real values to plot the upper bound
result in nearly undistinguishable results. The figure shows that
both the approximation and the upper bound are rather close to the
exact results. The T-Plot directly translates into a required
capacity overprovisioning: for instance, since
$GC_{CDF}^{T}(f,1.5)=0$, it is necessary to use at least 50\%
overprovisioning to guarantee that at least some traffic matrices
can be services (we remind that we used a pessimistic assumption
regarding traffic demands). Likewise, $GC_{CDF}^{T}(f,2.7)=1$, and
so a capacity of 2.7 would guarantee 100\% throughput. We will see
below that it is possible to guarantee 100\% using a much lower
overprovisioning, by deploying an optimized capacity allocation.

\begin{figure}
\centering
\subfigure[Global congestion CDF (over $\setA$)]{
\includegraphics
% [bb=0 920 1200 50,width=7cm, height=5cm]
[width = 0.66\columnwidth]
{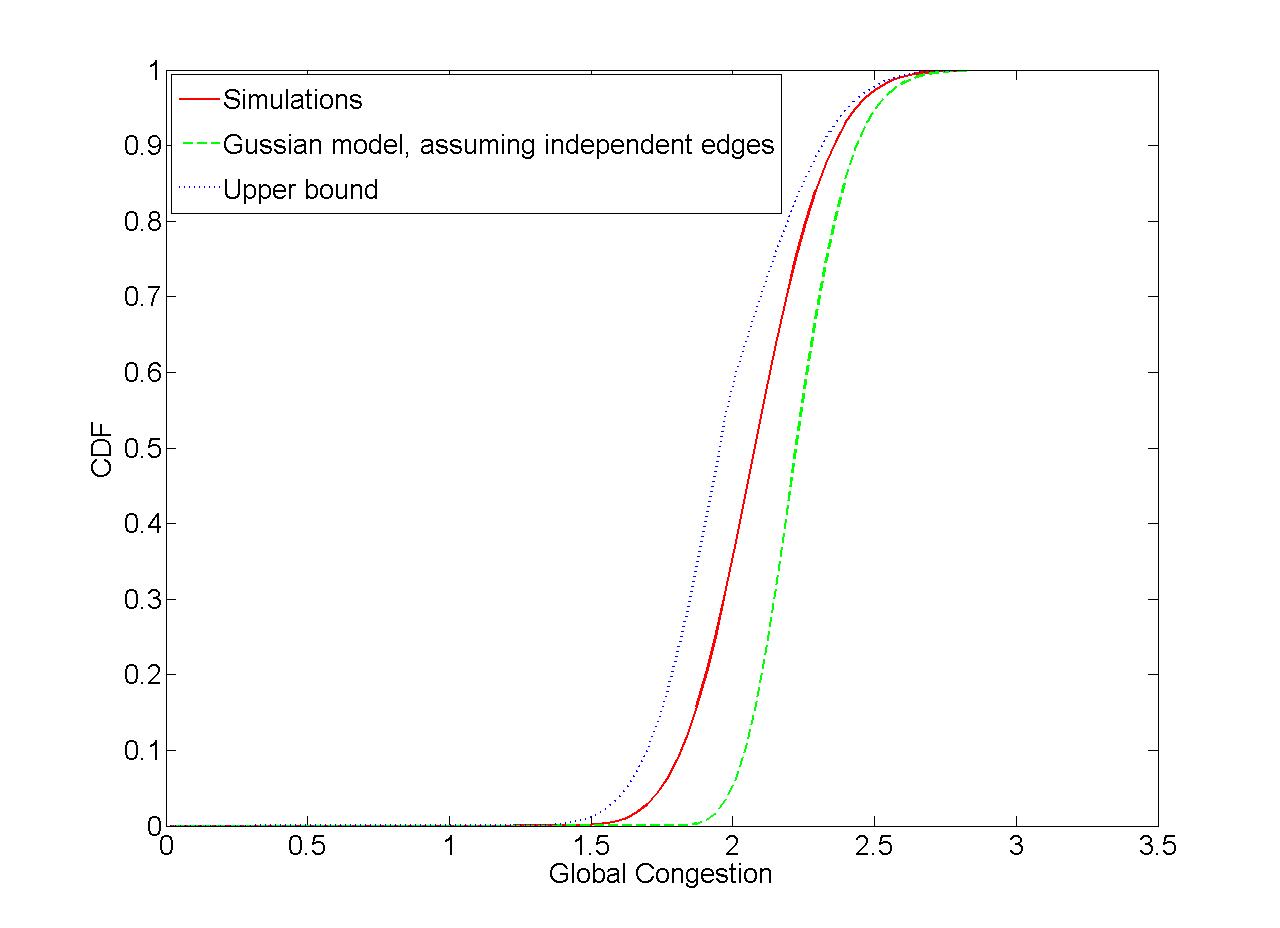}
\label{fig:Abi_vol_MEC_CDF}}

\subfigure[Global congestion CDF, for various Capacity Allocations (over $\setA$)] {
\includegraphics
% [bb=0 920 1200 50,width=7cm, height=5cm]
[width = 0.66\columnwidth]
{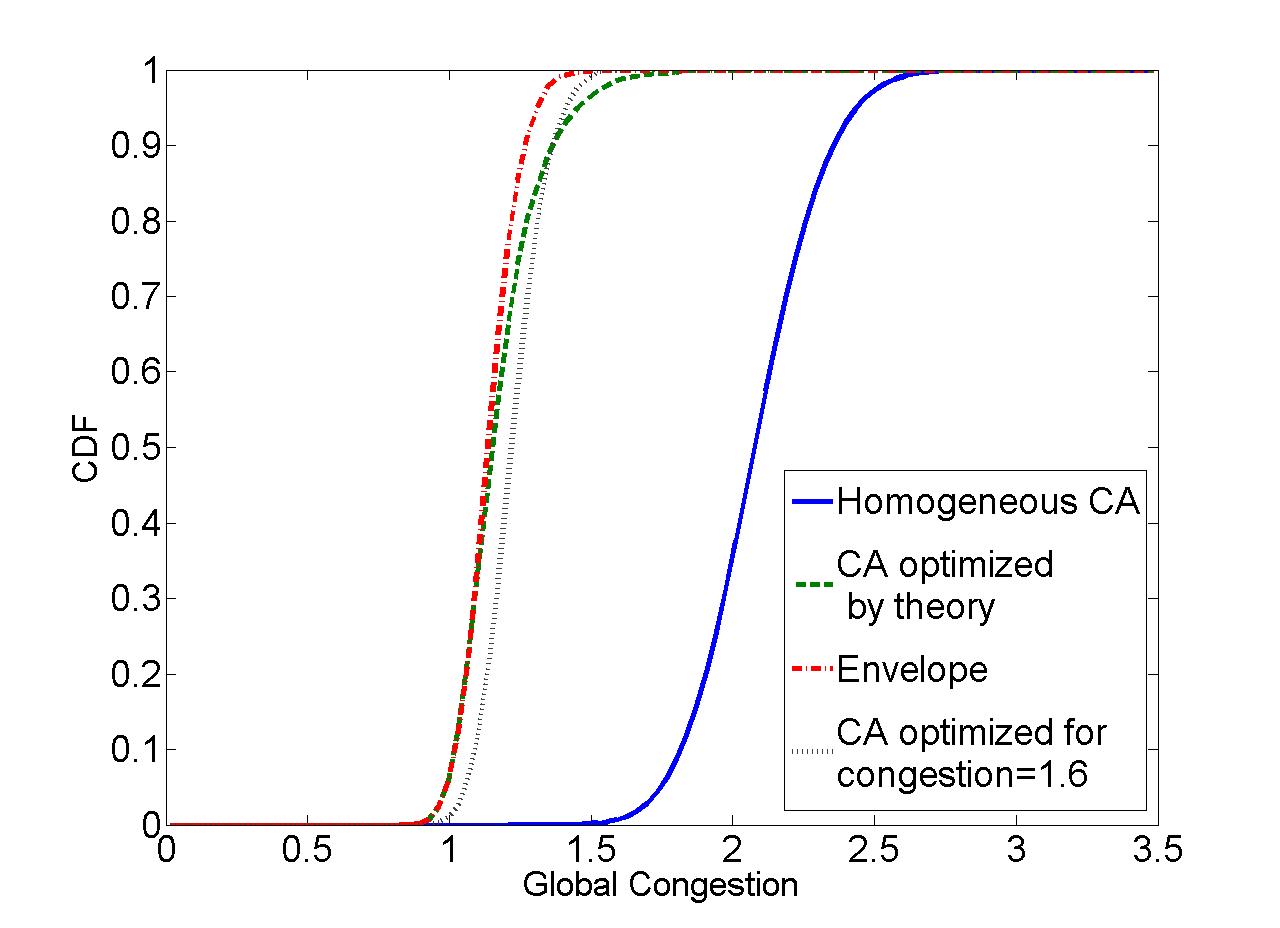}
\label{fig:Abi_DSS_CA_CDF}} 
\hfil \caption{Global congestion CDF: bounds and improvements}
\label{fig_sim_global}
\end{figure}

\section{Improved capacity allocation}

Until now, we realized every T-Plot "as-is", without doing any
optimization: we simply measured the distribution of the link load
distribution. We will now show that T-Plots can do more than just
\emph{measure}: they can also help \emph{optimize}.

Until now, we have used an homogeneous capacity allocation of the
Abilene backbone network. We assumed that the homogeneous Abilene
network contains 28 edges, each of unit capacity - i.e. total
capacity is 28 units.

We will now assume that capacity can be distributed differently
among the edges, so that the total allocation cost is simply equal
to the sum of the capacities. Using the simple capacity allocation
scheme presented in the Chapter~\ref{Sec8:Improving-the-capacity},
we force $\sum_{e\in E}c(e)=28$ in Equation (\ref{eq:Cap allocation
- sec 5}), and get $k=1.14$.

To examine the quality of our capacity allocation scheme, we compare
it to an optimized capacity allocation scheme, which was obtained
after extensive simulations. Given an optimization criterion, the
simulations include 10,000 iterations. At each iteration, a new
capacity allocation is taken at the neighborhood of the old one
using the Gaussian ball-walk distribution mentioned above (Chapter
\ref{sub:T-Set-representation}), and is only accepted if it fares
better. At each point, the throughput was calculated using 10,000
u.a.r.-generated $\setA$ matrices. (We also checked that starting
from different points yields the same result.)

But what is the optimization criterion? It turns out that optimizing
for different values in the global congestion CDF will
yield different solutions. Therefore, we performed different
optimizations on the whole spectrum of the global congestion CDF. For each value within this spectrum, we maximized the fraction of
matrices yielding a global congestion under this value. Thus, we got
a full envelope plot indicating an upper bound of the achievable
CDF.

This is better understood with Figure \ref{fig:Abi_DSS_CA_CDF},
which shows the CDF of the global congestion when using different
capacity allocation schemes. First, the homogeneous capacity
allocation performs poorly relative to the other capacity
allocations. Next, we ran the optimization process described above
for the single congestion value of $1.6$. While it obtains an
optimal allocation for that value, i.e., maximizes the fraction of
matrices with a global congestion under 1.6, it is clear that it
does not perform optimally elsewhere. Therefore, an allocation
scheme that is optimal at one point is not necessarily optimal
across the whole spectrum. For this reason, we also determine the
envelope plot of optimal values at each point of the spectrum (note
that each point of this envelope corresponds to a different scheme;
thus this envelope is an upper bound, rather than the performance of
a specific allocation scheme). Finally, we show the results of our
simple allocation. It clearly performs well, is close to the upper
bound, and on most points behaves better than the scheme that is
optimal at the congestion value of $1.6$. Further, it needs some
40\% less capacity than the homogeneous scheme to guarantee most
performance levels.

\section{Heterogeneous network}

We finally present a heterogenous Abilene model taking into account
the diverse capacity at each of its nodes. The links still run at
10Gbps links, but different nodes may initiate or receive traffic at
different maximum rates. We defined these maximum rates as equal to
the maximum rates over a one-day trace, which is publicly
available~\cite{key-Abi-traffic}. The T-Plot in
Figure~\ref{fig:Abi_H} shows that the global congestion is always
far below 1, i.e. Abilene can support all possible traffic demands
under this model. Note that the plot has lost the Gaussian behavior
found in the homogeneous case
(Figure \ref{fig:Abi_vol_MEC_PDF}).

As described in Section ~\ref{AppendixE}, the Gaussian model is
also applicable for heterogeneous networks. However, currently,
calculating the average and the variance of the edge congestion in
such networks requires computing $O(n^{4})$ different integrals, and
thus we don't present it here. More efficient Gaussian modeling of
heterogenous network is a subject for future research.

\begin{figure}
\centering
\includegraphics
% [bb=0 920 1200 50,width=7cm, height=5cm]
[width = 0.66\columnwidth]
{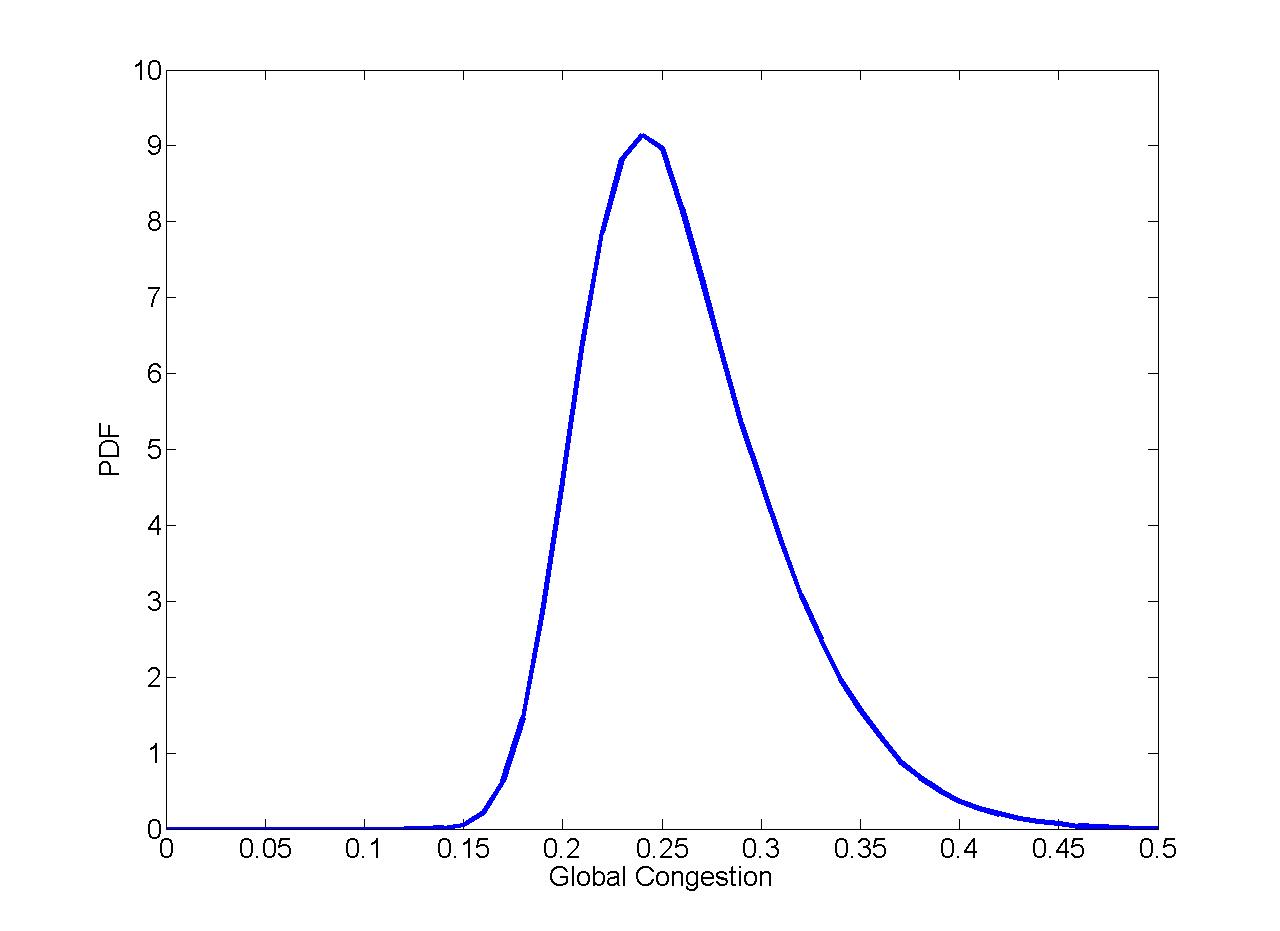}
\caption{PDF of the global congestion for the heterogeneous T-Set}
\label{fig:Abi_H}
\end{figure}

\chapter{Conclusion}

In this Thesis, we introduced the T-Plots, that can provide a
common foundation to quantify, design, optimize and compare traffic
engineering algorithms. We also showed that their computation is
\#P-Complete, but that they can sometimes be modeled as Gaussian,
providing a full link load distribution model using only two
variables. Further, we provided bounds that can be the basis of
strict throughput performance guarantees. We finally showed how
T-Plots can be used to develop a simple, yet efficient, capacity
allocation scheme. We believe and hope that this Thesis will
contribute to lay the ground to a common basis in future traffic
engineering research.

Further, even though we presented our work in the context of
backbone networks, it is clear that one can apply it to
interconnection networks (such as meshes, hypercubes or
bounded-degree graphs), and it would be interesting to study how it
relates to the work on oblivious routing in these topologies.

% that's all folks
%%%%%%%%@@@@@@@@@@@@@@@@@@@@@@@@@@@@@@@@@@@@@@@@@@@@@@@@@

% \bibliographystyle{plain}

{}


\begin{thebibliography}
\bibitem{key-Abi}Abilene backbone network. {\sf
https://en.wikipedia.org/wiki/Abilene\_Network}

\bibitem{key-Abi_is_10Gbps}Abilene network upgrade. {\sf
https://www.internet2.edu/news/detail/1978/}

\bibitem{key-Abi-traffic} Abilene aggregate traffic, daily numbers, accessed in Jan. 2007. {\sf
http://www.abilene.iu.edu/abilene/
maps-graphs/aggregate-traffic.html}

\bibitem{Agarwal}S. Agarwal, A. Nucci, and S. Bhattacharyya. Measuring the shared
fate of IGP engineering and interdomain traffic. In {\em Proc
International Conference on Network Protocols (ICNP)}, 2005.

\bibitem{PRE-PROVISIONING} M. Alicherry and R. Bhatia.
Pre-provisioning networks to support fast restoration with minimum
over-build, In {\em Proc. Infocom}, pages 164-175, 2004.

\bibitem{key-Uncertain_TMs} W. B. Ameur and H. Kerivin.
Routing of uncertain traffic demands. {\em Optimization and
Engineering,} 6:283-313, 2005.

\bibitem{key-COPING} D. Applegate, L. Breslau, and E. Cohen. Coping
with network failures: routing strategies for optimal demand
oblivious restoration. In {\em Proc. SIGCOMM}, pages 253-264, 2005.

\bibitem{CHANGING_TRAFFIC} D. Applegate and E. Cohen. Making intra-domain
routing robust to changing and uncertain traffic demands:
understanding fundamental tradeoffs. In {\em Proc. SIGCOMM}, pages
313-324, 2003.

\bibitem{OPT-OBL-POLY} Y. Azar, E. Cohen, A. Fiat, H. Kaplan,
and H. Racke. Optimal oblivious routing in polynomial time. In {\em
Proc. ACM Symposium on theory of Computing}, pages 383-388, ACM
Press, 2003.

\bibitem{key-01Perm} A. Ben-Dor and S. Halevi. Zero-one permanent
is \#P-Complete, a simpler proof. In {\em Proc. Israel Symposium on
Theory of Computing and Systems}, IEEE Press, 1993.

\bibitem{PRACTICAL-OPT-OBL-POLY} M. Bienkowski, M. Korzeniowsk and H. Racke.
A practical algorithm for constructing oblivious routing, In {\em
Proc. ACM symposium on Parallel algorithms and architectures}, pages
24 - 33, 2003.

\bibitem{key-Birkhoff}
G. Birkhoff. Tres observaciones sobre
el algebra lineal. {\em Univ. Nac. Tucuman Rev. Ser. A}, 5:147-151,
1946.

\bibitem{Chambers}
J. Chambers, W. Cleveland, B. Kleiner, and P. Tukey. Graphical
methods for data analysis. Wadsworth, 1983.

\bibitem{key-vol_of_poly_of_SD} C. S. Chan and D.P. Robbins. On
the volume of the polytope of doubly stochastic matrices. {\em
Experimental Mathematics}, 8:291-300, 1999.

\bibitem{key-Dijkstra}
E. W. Dijkstra. A note on two problems
in connexion with graphs. {\em Numerische Mathematik}, 1:269-271,
1959.

\bibitem{Duffield}
N. Duffield, P. Goyal, and A. Greenberg. A flexible model for
resource management in virtual private networks. In {\em Proc.
SIGCOMM}, 1999.

\bibitem{key-Typical_vs_WC}
N. Dukkipati, Y. Ganjali, and R. Zhang-Shen. Typical versus worst
case design in networking. HotNets-IV, College Park, November 2005.

\bibitem{key-Vol_is_P-C}M. E. Dyer. and A. M. Frieze.
On the complexity of computing the volume of a polyhedron. {\em SIAM
Journal on Computing}, 17(5):967-974, 1988.

\bibitem{Elekes}
G. Elekes. A geometric inequality and the complexity of measuring
the volume. \emph{Discrete Comput. Geom.}, 1:289-292, 1986.

\bibitem{Esseen}
C. G. Esseen. Fourier analysis of distribution functions: a
mathematical study of the Laplace-Gaussian law. \emph{Acta Math.},
77(1):1-125, 1945.

\bibitem{STRETCH_N_LB} J. Gao, L. Zhang.
Tradeoffs between stretch factor and load balancing ratio in routing on growth restricted graphs,
In {\em Proc. ACM symposium on Principles of distributed computing}, pages: 189-196, 2004.

\bibitem{Gumbel}
E. J. Gumbel. Multivariate extremal distributions. \emph{Bull. Inst.
Internat. de Statistique},  37: 471-475, 1960.

\bibitem {Obl-general}A. Gupta, M. T. Hajiaghayi, and H. Racke,
Oblivious network design, In {\em Proc. ACM-SIAM Symposium on
Discrete Algorithms (SODA)}, pages 970-979, 2006.

\bibitem{BOUNDS} M. T. Hajiaghayi, R. D. Kleinberg, T. Leighton, and H. Racke.
Oblivious routing on node-capacitated and directed graphs, In {\em
Proc. ACM-SIAM symposium on Discrete algorithms}, pages: 782-790,
2005.

\bibitem{NEW-BOUNDS} M. T. Hajiaghayi, R. D. Kleinberg, T. Leighton, and H. Räcke.
New lower bounds for oblivious routing in undirected graphs, In {\em
Proc. ACM-SIAM symposium on Discrete algorithms}, pages: 918-927.
2006.

\bibitem{key-PATHNECK} N. Hu, L. Li, Z. Mao, P. Steenkiste, and J.
Wang. Locating internet bottlenecks: algorithms, measurements, and
implications. In {\em Proc. SIGCOMM}, pages 41-54, 2004.

\bibitem{key-MIRA} K. Kar, M. Kodialam, and T.V. Lakshman. Minimum
interference routing of bandwidth guaranteed tunnels with MPLS
traffic engineering applications. {\em IEEE Journal on Selected
Areas in Communications} 18(12):2566-2579, 2000.

\bibitem{OBL_RESTORATION} M. Kodialam, T. V. Lakshman, and S. Sengupta.
A simple traffic independent scheme for enabling restoration
oblivious routing of resilient connections, In {\em Proc. Infocom},
pages 2329-2340, 2004.

\bibitem{Kodialam}
M. Kodialam, T. V. Lakshman, and S. Sengupta. Efficient and robust
routing of highly variable traffic, In {\em Proc. Workshop on Hot
Topics in Networks (HotNets-III)}, 2004.

\bibitem{key-Lillie}H. Lilliefors. On the Kolmogorov-Smirnov test
for normality with mean and variance unknown. {\em Journal of the
American Statistical Association,} 62:399-402, 1967.

\bibitem{VLB_DIRECT}
H. Liu and Rui Zhang-Shen. On direct routing in the Valiant
load-balancing architecture, IEEE Globecom 2005, St. Louis, MO,
November 2005

\bibitem{Roughan}
M. Roughan, M. Thorup, and Y. Zhang. Traffic engineering with
estimated traffic matrices. In {\em Proc. Internet Measurement
Conference}, 2003.

\bibitem{key-O1TURN}D. Seo, A. Ali, W. T. Lim, N. Rafique, and M.
Thottethodi. Near-optimal worst-case throughput routing for
two-dimensional mesh networks. In {\em Proc. International Symposium
on Computer Architecture}, pages 432-443, June 2005.

\bibitem{key-GOAL}A. Singh, W.J. Dally, A.K. Gupta, and B. Towles.
GOAL: a load-balanced adaptive routing algorithm for torus networks.
In {\em Proc. 30th Annu. international Symposium on Computer
Architecture}, pages 194-205, 2003.

\bibitem{key-DOR}H. Sullivan and T. R. Bashkow. A large scale, homogeneous,
fully distributed parallel machine. In {\em Proce. Symposium on
Computer architecture}, pages 105-117. ACM Press, 1977.

\bibitem{VPNS} T. Takeda, R. Matsuzaki, I. Inoue, and S. Urushidani.
Network design scheme for virtual private network services {\em
IEICE Transactions on Communications}, E89-B.1:3046-3054, 2006.

\bibitem{key-WC_traffic} B. Towles and W. J. Dally. Worst-case
traffic for oblivious routing functions. {\em ACM Symposium on
Parallel Algorithms and Architectures}, 2002.

\bibitem{key-TP_centric}B. Towles, W. J. Dally, and S. Boyd. Throughput-centric
routing algorithm design. In {\em Proc. 15th Annu. ACM Symposium on
Parallel Algorithms}, pages 200-209, 2003.

\bibitem{key-Val} L. G. Valiant. The complexity of enumeration
and reliability problems. {\em SIAM J. Comput.}, 8(3):410-421, 1979.

\bibitem{key-Random_walk} S. Vempala. Geometric random walks:
a survey. {\em MSRI volume on Combinatorial and Computational
Geometry}, 2005.

\bibitem{key-Vinci}Vinci, {\sf
http://www.lix.polytechnique.fr/Labo/ Andreas.Enge/Vinci.html}

\bibitem{key-VN}J. von Neumann. A certain zero-sum two-person
game equivalent to the optimal assignment problem. {\em
Contributions to the Theory of Games}, Princeton Univ. Press,
2:5-12, 1953.

\bibitem{key-COPE} H. Wang, H. Xie, L. Qiu, Y. R. Yang, Y. Zhang,
and A. Greenberg. COPE: traffic engineering in dynamic networks. In
{\em Proc. SIGCOMM}, pages 99-110, 2006.

\bibitem{Wang1}
Z. Wang, Y. Wang, and L. Zhang. Internet traffic engineering without
full mesh overlaying. In {\em Proc. Infocom}, 2001.

\bibitem{Ye}
T. Ye, H. Kaur, S. Kalyanaraman, K. Vastola, and S. Yadav.
Optimization of OSPF weights using online simulation. In {\em Proc.
International Workshop on Quality of Service (IWQoS)}, 2002.

\bibitem{Zhang1}
C. Zhang, Z. Ge, J. Kurose, Y. Liu, and D. Towsley. Optimal routing
with multiple traffic matrices: Tradeoff between average case and
worst case performance. In {\em Proc. International Conference on
Network Protocols (ICNP)}, 2005.

\bibitem{Zhang2}
C. Zhang, Y. Liu, W. Gong, J. Kurose, R. Moll, and D. Towsley. On
optimal routing with multiple traffic matrices. In {\em Proc.
Infocom}, 2005.

\bibitem{Zhang3}
Y. Zhang and Z. Ge. Finding critical traffic matrices. {\em
Proc. of DSN 05}, Yokohama, Japan, June 2005.

\bibitem{VLB-HOMO} R. Zhang-Shen and N. McKeown.
Designing a predictable Internet
backbone network. In {\em Proc. Workshop on Hot Topics in Networks
(HotNets-III)}, 2004.

\bibitem{key-VLB} R. Zhang-Shen and N. McKeown. Designing
a predictable Internet backbone with Valiant load-balancing. {\em
Proc. International Workshop on Quality of Service (IWQoS)}, 2005.

\bibitem{Noc-Conf} I. Cohen, O. Rottenstreich, I. Keslassy. 
Statistical Approach to Networks-on-Chip. {\em
Proc. ACM/IEEE NoCS}, pp. 171-180, 2008.

\bibitem{Noc-Conf-TR} I. Cohen, O. Rottenstreich, I. Keslassy. 
Statistical Approach to Networks-on-Chip ((extended version). {\em
Proc. ACM/IEEE NoCS}, Vol. 1. Technical Report TR08, 2008.

\bibitem{Noc-Journal} I. Cohen, O. Rottenstreich, I. Keslassy. 
Statistical Approach to Networks-on-Chip. {\em
IEEE Transactions on Computers (ToC)}, 59(6):748-761, 2010.

\bibitem{chudzikiewicz2013resources} Chudzikiewicz, J., Zielinski, Z. Resources placement in the 4-dimensional fault-tolerant hypercube processors network. 
{\em  FedCSIS Position Papers}, pp. 69--74, 2013.

\bibitem{Beyond-Disco-Balls} 
Kassing, S., Valadarsky, A., Shahaf, G., Schapira, M., Singla, A.
Beyond fat-trees without antennae, mirrors, and disco-balls. 
in {\em Proc. Conference of the ACM Special Interest Group on Data Communication},
pp.281--294, 2017.

\bibitem{Qmist-Conf}
Cohen I., Scalosub G. Queueing in the Mist: Buffering and Scheduling with Limited Knowledge. in {\em Proc. ACM/IEEE IWQoS.
}, pp.1-6, 2017.

\bibitem{Qmist-Journal}
Cohen I., Scalosub G. Queueing in the Mist: Buffering and Scheduling with Limited Knowledge. {\em Computer Networks}, 
147:204-220, 2018.

\bibitem{APSR-Infocom20-poster}
Cohen, I., Einziger, G., Goldstein, M., Sa'ar, Y., Scalosub, G., & Waisbard, E. Parallel vm placement with provable guarantees. in {\em Proc. IEEE Infocom WKSHP}, pp. 1298-1299, 2020.

\bibitem{APSR-IFIP}
Cohen, I., Einziger, G., Goldstein, M., Sa'ar, Y., Scalosub, G., & Waisbard, E. Parallel vm deployment with provable guarantees. in {\em Proc. IEEE IFIP Networking} (pp. 1-9), 2021.

\bibitem{FN-Aware-Conf}
Cohen, I., Einziger, G. and Scalosub, G., 2021. On the power of false negative awareness in indicator-based caching systems. in {\em Proc. IEEE ICDCS}, 2021.

\bibitem{Accs-Strategies-Conf}
Cohen I., Einziger, G., Friedman R., Scalosub G. Access Strategies for Network Caching. in {\em Proc. IEEE Infocom}, pp.28-36, 2019.

\bibitem{Accs-Strategies-ToN}
Cohen I., Gil E., Friedman R., Scalosub G. Access Strategies for Network Caching. in {\em IEEE Transactions on Networking}, 2020.

\bibitem{CAB}
Cohen I., Gil E., Scalosub G. Self-adjusting Advertisement of Cache Indicators
with Bandwidth Constraints. in {\em Proc. IEEE Infocom}, 2021.
\end{thebibliography}
\end{document}